\documentclass{lmcs}
\pdfoutput=1

\usepackage{lastpage}
\lmcsdoi{15}{1}{22}
\lmcsheading{}{\pageref{LastPage}}{}{}%
{Feb.~14,~2018}{Mar.~05,~2019}{}

\usepackage{amsmath, amsfonts}
\usepackage{amssymb, latexsym}
\usepackage{amsthm}
\usepackage{stmaryrd}
\usepackage{cite}

\usepackage{hyperref}

\usepackage{tikz-cd}
\tikzset{close/.style={near start,outer sep=-10pt}}

\newcommand{\Bin}{\left\{ 0,1 \right\}}
\newcommand{\Pow}[1]{\mathcal{P}(#1)}
\DeclareMathOperator{\amp}{\,\&\,}
\DeclareMathOperator{\imp}{\,\rightarrow\,}
\DeclareMathOperator{\defeqiv}{\stackrel{\textup{def}}{\iff}}
\DeclareMathOperator{\defeql}{\stackrel{\textup{def}}{\  =\  }}
\DeclareMathOperator{\elm}{\epsilon}
\DeclareMathOperator{\backelm}{\text{\reflectbox{$\epsilon$}}}
\newcommand{\dotminus}{\mathbin{\ooalign{\hss\raise.5ex\hbox{$\cdot$}\hss\crcr$-$}}}

\DeclareMathOperator{\cov}{\mathbin{\lhd}}
\DeclareMathOperator{\bcov}{\mathrel{\blacktriangleleft}}
\DeclareMathOperator{\pos}{\ltimes}

\DeclareMathOperator{\meets}{\between}
\newcommand{\sat}{\mathcal{A}}
\newcommand{\ext}{\mathsf{ext}\,}
\newcommand{\rest}{\mathsf{rest}\,}
\newcommand{\Ext}{\mathsf{Ext}}
\newcommand{\downset}{\mathord{\downarrow}}
\newcommand{\Image}[1]{{#1}_{*}}
\newcommand{\im}{\mathrm{Im}}
\newcommand{\Nat}{\mathbb{N}}

\newcommand{\FSeq}{\Nat^{*}}
\newcommand{\BSeq}{\Bin^{*}}
\newcommand{\FBaire}{\mathcal{B}}
\newcommand{\FReal}{\mathcal{R}}
\newcommand{\FUInt}{\mathcal{I}[0,1]}
\newcommand{\UInt}{\left[ 0,1 \right]}
\newcommand{\Discrete}[1]{\mathcal{D}#1}
\newcommand{\FNat}{\mathcal{N}}

\newcommand{\nil}{\mathbf{\mathit{nil}}}

\newcommand{\cons}[2]{{#1} * {#2}}
\newcommand{\lh}[1]{\lvert #1 \rvert}
\newcommand{\head}[1]{\mathsf{head}(#1)}
\newcommand{\tail}[1]{\mathsf{last}(#1)}
\newcommand{\Cont}{\mathrm{PoPC}}

\newcommand{\ContUInt}{\mathrm{PoPC}_{\FUInt}}

\newcommand{\ContBN}{\mathrm{PoPC}_{\FBaire,\FNat}}
\newcommand{\ContBND}{\mathrm{D}\text{-}\ContBN}
\newcommand{\ContSigmaOne}{\Sigma^{0}_{1}\text{-}\ContBN}
\newcommand{\ContPiOne}{\Pi^{0}_{1}\text{-}\ContBN}
\newcommand{\LPO}{\mathrm{LPO}}
\newcommand{\LLPO}{\mathrm{LLPO}}
\newcommand{\Pt}[1]{\mathit{IP}t(#1)}
\newcommand{\Ip}{\mathcal{I}p}
\newcommand{\Fs}{\mathrm{Fs}}
\newcommand{\BISigmaOne}{\mathrm{\Sigma^{0}_{1}\text{-}\mBI}}
\newcommand{\BIPiOne}{\mathrm{\Pi^{0}_{1}}\text{-}\mBI}
\newcommand{\BI}{\mathrm{BI}}
\newcommand{\mBI}{\BI_{\mathbf{M}}}
\newcommand{\dBI}{\BI_{\mathbf{D}}}

\DeclareMathOperator{\id}{\mathrm{id}}
\DeclareMathOperator{\dom}{\mathrm{dom}}

\newcommand{\CS}{{\textrm{\textup{CS}}}}
\newcommand{\CSpa}{\textbf{\textrm{\textup{CSpa}}}}
\newcommand{\PTop}{\textbf{\textrm{\textup{PTop}}}}

\theoremstyle{plain}
\newtheorem{theorem}{Theorem}[section]
\newtheorem{proposition}[theorem]{Proposition}
\newtheorem{lemma}[theorem]{Lemma}
\newtheorem{corollary}[theorem]{Corollary}
\theoremstyle{definition}
\newtheorem{definition}[theorem]{Definition}
\newtheorem{remark}[theorem]{Remark}

\newtheorem{notation}[theorem]{Notation}
\newtheorem{example}[theorem]{Example}
\numberwithin{equation}{section}

\begin{document}
\title[The principle of pointfree continuity]{The principle of pointfree continuity}

\author[T. Kawai]{Kawai Tatsuji}       
\address{Japan Advanced Institute of Science and Technology
1-1 Asahidai, Nomi, Ishikawa 923-1292, Japan}     
\email{tatsuji.kawai@jaist.ac.jp}  

\author[G. Sambin]{Giovanni Sambin}     
\address{Dipartimento di Matematica, Universit\`{a} di Padova, via Trieste 63, 35121 (PD) Italy}      
\email{sambin@math.unipd.it}  


\begin{abstract}
  In the setting of constructive pointfree topology, we introduce a
  notion of continuous operation between pointfree topologies and the
  corresponding principle of pointfree continuity. An operation
  between points of pointfree topologies is continuous if it is
  induced by a relation between the bases of the topologies;
  this gives a rigorous condition for Brouwer's continuity principle to hold.
  The principle of pointfree continuity for pointfree topologies
  $\mathcal{S}$ and $\mathcal{T}$ says that any relation which
  induces a continuous operation between points is a morphism from
  $\mathcal{S}$ to $\mathcal{T}$.
  The principle holds under the assumption of bi-spatiality of
  $\mathcal{S}$.  When $\mathcal{S}$ is the formal Baire space or the
  formal unit interval and $\mathcal{T}$ is the formal topology of natural
  numbers, the principle is equivalent to spatiality of the
  formal Baire space and formal unit interval, respectively.
  Some of the well-known connections between spatiality, bar
  induction, and compactness of the unit interval are recast
  in terms of our principle of continuity.

  We adopt the Minimalist Foundation as our constructive foundation,
  and positive topology as the notion of pointfree topology.
  This allows us to distinguish ideal objects from constructive
  ones, and in particular, to interpret choice sequences as points of
  the formal Baire space.
\end{abstract}
\maketitle

\section{Introduction}\label{sec:Introduction}
In a number of writings, Brouwer analysed functions from
choice sequences to the natural numbers, and claimed that every total
function on choice sequences is continuous in a very strong sense (see
e.g.\ Brouwer \cite{BrouwerDomainsofFunctions}). 
As a corollary, he obtained the continuity theorem of real numbers: every total function on the unit interval is uniformly
continuous.
In more recent accounts of Brouwer's intuitionism (e.g.\ Kleene and
Vesley~\cite{KleeneVesley}), it is common to decompose Brouwer's
analysis into two principles: continuity principle and bar
induction.\footnote{There are many variations of these principles,
depending on the strength of continuity and complexity of bars. The
versions of continuity principle and bar induction that we present
here are deliberately chosen to suite our aims.} Continuity principle
says that every total function from choice sequences to the natural
numbers is pointwise continuous, while bar induction says that every
monotone bar is an inductive bar (cf.\ Troelstra and van Dalen
\cite[Chapter 4]{ConstMathI}).

The aim of this paper is to show that continuity principle can be maintained if
we take seriously the constructive and pointfree approach to topology
\cite{Sambin:intuitionistic_formal_space,Sambin_BP_book}, while bar
induction is the principle which is exactly needed to bridge the gap
between pointwise continuity and pointfree continuity.

As to the first claim, in the constructive pointfree topology, a
point appears as a set of its formal neighbourhoods.
Therefore, it is reasonable
to assume that every constructively defined operation on the
collection of points must be induced by a relation between the formal
neighbourhoods of the relevant pointfree topologies.
Actually, one of
the main points of this paper is that every pointwise continuous operation
between points of pointfree topologies is induced
by a relation between formal neighbourhoods, and conversely if a
relation between formal neighbourhoods gives rise to a total
operation between points, then such operation is
automatically pointwise continuous.  Hence, we obtain continuity for free if we
start from a relation between formal neighbourhoods; see Theorem~\ref{thm:PWCont}.

Given two pointfree topologies $\mathcal{S}$ and $\mathcal{T}$, 
it is then natural to ask whether a relation
from formal neighbourhoods of $\mathcal{S}$
to those of $\mathcal{T}$ which gives rise to a total operation
between points is actually a pointfree 
map from $\mathcal{S}$ to $\mathcal{T}$.
This is not necessarily the case unless the topology on points of
$\mathcal{S}$ coincides with $\mathcal{S}$. Hence, we formulate this
transition to a pointfree map as a
principle, called \emph{the principle of pointfree continuity}
for $\mathcal{S}$ and $\mathcal{T}$
($\Cont_{\mathcal{S},\mathcal{T}}$); see Definition \ref{def:Cont} for the precise formulation.
We discuss some sufficient conditions on $\mathcal{S}$ and
$\mathcal{T}$ under which $\Cont_{\mathcal{S},\mathcal{T}}$ holds in
Section~\ref{sec:Cont}.

The above claim and the principle of pointfree continuity
make sense for arbitrary pointfree topologies --- not just for the Baire
space and the natural numbers for which continuity principle and bar
induction are usually formulated.
However, the principle $\ContBN$
for the formal (i.e.\ pointfree) Baire space $\FBaire$ and the formal
topology of natural numbers $\FNat$ is equivalent to 
 monotone bar induction,
and hence to spatiality of the formal Baire space.
Moreover, some well-known notions such as neighbourhood function and
decidable bar induction \cite[Chapter 4]{ConstMathI} naturally
arise when we consider a variant of $\ContBN$ formulated with
respect to a restricted class of relations
from the formal neighbourhoods of $\FBaire$ to $\FNat$; see
Section~\ref{sec:ContBaire}.
Furthermore, the principle $\Cont_{\FUInt, \FNat}$ for the formal unit
interval $\FUInt$ and the formal topology of natural numbers $\FNat$  is shown to
be equivalent to Heine--Borel covering theorem, and thus to spatiality
of the formal topology of real numbers; see Section~\ref{sec:ContReal}.

\subsection*{Remark on foundations}
In this paper, we adopt the Minimalist
Foundation~\cite{TowardMinimalist,MaiettiMinimalist} as our constructive foundation and positive
topology \cite{Sambin_BP_book} as the notion of pointfree topology
(see Section~\ref{sec:PTop}). 
For the details of the Minimalist Foundation, we refer the reader to
\cite{TowardMinimalist, MaiettiMinimalist} 
or the first two chapters
of the forthcoming book by the second author~\cite{Sambin_BP_book}.
It should be noted, however, that the results in
Section~\ref{sec:ContBaire} and Section~\ref{sec:ContReal} do not make
essential use of the structure of positive topology, and thus they can be
understood in the setting of formal topology~\cite{Sambin:intuitionistic_formal_space}.
 
In what follows, we elaborate on why working in the Minimalist
Foundation could be relevant to this work.
The Minimalist Foundation distinguishes collections from sets and
logic from type theory. In particular, points of a positive
topology usually form a collection rather than a set, and hence we can
regard a point of a positive topology as an ideal object.
For example, points of the formal Baire space, which are equivalent to
functions between the natural numbers, form a collection rather than a
set. This is due to the separation of logic from type theory, which
allows us to keep the logical notion of function distinct from the type
theoretic notion of constructive operation. Because of this, we can view
a point of the formal Baire space as a choice sequence which is not
necessarily lawlike (see also the discussion following Proposition
\ref{col:CSasIPt}). 

It is clear from Brouwer's writings that not only the notion of choice
sequence but also that of lawlike operation (which he called
algorithm in \cite{BrouwerDomainsofFunctions}) plays a crucial role in
his analysis of continuity on the Baire space. Hence, any satisfactory
account of Brouwer's intuitionism requires both the notion of choice sequence
and that of lawlike sequence which are kept separate.
The Minimalist Foundation can serve as a practical foundation of
intuitionism where we can talk about choice sequences as a
``figure of speech''. For example, it might be possible to  postulate
bar induction on choice sequences
while maintaining the view that every type
theoretic operation is lawlike.\footnote{Another candidate for a
practical foundation of Brouwer's intuitionism is the theory of
choice sequences $\CS$ by Kreisel and
Troelstra \cite{KreiselTroelstra}. It should be noted, however, that
the treatment of choice sequences in the Minimalist Foundation is
quite different from the analytic approach of $\CS$ where conceptual
analysis of an individual choice sequence plays a central role. But
the treatment in the Minimalist Foundation  may be more coherent and easier to understand, since it
arises naturally from the fundamental distinction between collections
and sets, and logic and type theory.}

\begin{notation}
We adopt notations which illustrate distinctions between sets,
collections, and propositions in the Minimalist Foundation.
If $a$ is an element of a set $S$, we write $a \in S$, and if $a$ is an
element of a collection $\mathcal{C}$, we write $a \colon \mathcal{C}$.
A subset $U(x)$ of a set $S$ (written $U \subseteq S$) is a
propositional function on $S$ with at most one variable $x \in S$. If
$a \in S$ is an element of a subset $U \subseteq S$, that is $U(a)$ is
true, then we write $a \elm U$.
Two subsets are said to be equal if they have the same elements.

The collection of subsets of a set $S$ (i.e.\ the power of a set) is
denoted by $\Pow{S}$. Note that $\Pow{S}$ is not a set except when $S$
is empty.  We say ``\emph{impredicatively}'' to 
mean that we temporarily assume that
the power of a set is a set.

A relation from $X$ to $S$ is a propositional function with two
arguments, one in $X$ and one in $S$.  Equivalently, a relation from
$X$ to $S$ is a subset of the
cartesian product $X \times S$. 
Every relation $s \subseteq X \times S$ determines the image operation 
$\Image{s} \colon \Pow{X} \to \Pow{S}$ defined by
\[
  \Image{s}(D) \defeql \left\{ a \in S \mid \left( \exists x \elm D
    \right) x \mathrel{s} a \right\}
\]
for each subset $D \subseteq X$. We usually write 
$\Image{s}D$ for $\Image{s}(D)$ and $\Image{s} x$ for $\Image{s}
\left\{ x \right\}$.
Following the usual mathematical convention, we drop the
subscript ``${}_{*}$'' and simply write $sD$  and $sx$ whenever doing this
does not cause confusion. The inverse image operation $\Image{s^{-}}
\colon \Pow{S} \to \Pow{X}$ (often written simply $s^{-}$) is just the image operation of the
opposite relation $s^{-} \subseteq S \times X$ of $s$.

Lastly, since the logic of the Minimalist Foundation is intuitionistic, we distinguish 
between
inhabited subsets
and non-empty subsets. To this end, it is
convenient to use the following notation:
\[
  U \meets V \defeqiv \left( \exists a \in S \right) a \elm U \cap V.
\]
\end{notation}

\section{Continuity principle for positive topologies}\label{sec:ContPtop}
The Minimalist Foundation, in particular its distinction
between sets and collections, leads us to introduce 
two different notions of topology which replace the classical notion
of topological space.
When points form a set, we reach the notion of concrete space.
When points form a collection, it is more appropriate
to understand them as ideal points of a pointfree structure, which we
call positive topology.

We briefly review these notions in the next two subsections,  
and refer the reader to \cite{Sambin:some_points_in_FTop} or the forthcoming book
\cite{Sambin_BP_book} for further details.

\subsection{Concrete spaces}\label{sec:CSpa}
The first is the pointwise notion of topological space.
\begin{definition}\label{def:ConcreteSpace}
  A \emph{concrete space} is a triple $\mathcal{X} = (X, \Vdash, S)$ where
  $X$ and $S$ are sets and $\Vdash$ is a relation from $X$ to
  $S$ satisfying
  \begin{enumerate}[(B1)]
    \item $\ext a \cap \ext b = \ext(a \downarrow b)$
    \item $X = \ext S$
  \end{enumerate}
  for all $a,b \in S$, where
  \begin{align*}
  \ext a &\defeql \Vdash^{-}a  \\
  \ext U &\defeql \Vdash^{-}U = \bigcup_{a \elm U} \ext a\\
    a \downarrow b &\defeql 
    \left\{ c \in S \mid c \cov_{\mathcal{X}} \{a\} \amp c
    \cov_{\mathcal{X}} \{b\}\right\}\\
    a \cov_{\mathcal{X}} \{b\}&\defeqiv \ext a \subseteq \ext b
  \end{align*}
    for all $a,b \in S$ and $U \subseteq S$.
    The notation $\downarrow$ is extended to subsets by 
    \[
      U \downarrow V \defeql \bigcup_{a \elm U, b \elm
      V} a \downarrow b
    \]
    for all $U,V \subseteq S$.%
    \footnote{The notations $\cov_{\mathcal{X}}$ and $\downarrow$ are
    instances of the same notations that appear in Definition~\ref{def:PositiveTop}. This will become clear when we define  
    the notion of representable topology (cf.\
    \eqref{eq:Representable}).}
\end{definition}
Conditions (B1) and (B2) say that the subsets of the form $\ext a$
constitute a base for a topology on $X$. Thus, a concrete space is a
set equipped with an explicit set-indexed base.

\begin{definition}
  Let $\mathcal{X} = (X, \Vdash, S)$ and $\mathcal{Y} = (Y, \Vdash',
  T)$ be concrete spaces. 
  A \emph{relation pair} $(r,s)$ from $\mathcal{X}$ to $\mathcal{Y}$ is a
  pair of relations $r \subseteq X \times Y$ and $s \subseteq S \times T$ such that
  \[
    \Vdash' \mathop{\circ} r \mathrel{=}  s \mathop{\circ} \Vdash,
  \]
  where $\circ$ is composition of relations.
  A relation pair $(r,s) \colon \mathcal{X} \to \mathcal{Y}$
  is said to be \emph{convergent} if
  \begin{enumerate}[(C1)]
    \item $\ext\left(s^{-} a \downarrow
      s^{-} b \right) = \ext \, s^{-}(a \downarrow' b)$
    \item $\ext\, S  = \ext\, s^{-}T$
  \end{enumerate}
  for all $a,b \in T$. 
  Two convergent relation pairs $(r,s), (r',s')
  \colon \mathcal{X} \to \mathcal{Y}$ are defined to be \emph{equal} if
  \begin{equation}\label{eq:EqualityCSpa}
    \Vdash' \mathop{\circ} r \mathop{=}  \Vdash' \mathop{\circ} r'.
  \end{equation}
\end{definition}
The collections of concrete spaces and convergent relation pairs form a
category $\CSpa$. The identity on a concrete space is 
the pair of identity relations. Composition of two
convergent relation pairs is the coordinate-wise
composition of relations. It is easy to see that 
composition respects equality on relation pairs. 
By exploiting the way in which morphisms are
defined in $\CSpa$, one can prove that $\CSpa$ is impredicatively
equivalent to the category of \emph{weakly sober} topological spaces and continuous
functions.\footnote{The following remark assumes that the reader 
  is familiar with locale theory \cite{johnstone-82} and has read
  Section 2.2 and 2.3 of this paper. A concrete space $\mathcal{X}$  is
  \emph{weakly sober} if for every ideal point $\alpha$ of
  $\mathcal{S}_{\mathcal{X}}$, there exists a unique $x \in X$
  such that $\alpha = \Diamond x$. 
  Classically, weak sobriety is equivalent to 
  sobriety, but intuitionistically it is strictly weaker. This is because the
  ideal points correspond to a certain subclass of completely prime
  filters of opens of $\mathcal{X}$;
  see Aczel and Fox \cite{AczelFoxSeparation} (n.b.\ ideal points are called
  strong ideal points in \cite{AczelFoxSeparation}).

  One can show that the category of
  weakly sober concrete spaces and \emph{continuous functions} is
  impredicatively equivalent to $\CSpa$. Indeed, the former category
  can be easily embedded into $\CSpa$. On the other hand, every concrete
  space $\mathcal{X}$ is impredicatively isomorphic to
  $\Ip(\mathcal{S}_{\mathcal{X}})$, which is weakly sober.
  Then, by defining the notion of weak sobriety for topological spaces
  in terms of completely prime filters that
  correspond to ideal points, the claimed equivalence is obtained. A
  detailed proof appears in~\cite{Sambin_BP_book}.}

\subsection{Positive topologies}\label{sec:PTop}
The pointfree notion of topology arises from abstraction of the
structure induced on the base (or its index set thereof) of a concrete
space; see \eqref{eq:Representable} below.
\begin{definition}\label{def:PositiveTop}
  A \emph{positive topology} is a triple $(S, \cov,
  \pos)$ where $S$ is a set, called the \emph{base}, and $\cov$ 
  and $\pos$ are relations from $S$ to $\Pow{S}$ such that
  \begin{gather*} \frac{a \elm U}
    {a \cov U}\,
    (\text{reflexivity})
    \quad
    \frac{a \cov U \quad U \cov V}
    {a \cov V}\,
    (\text{transitivity})
    \quad
    \frac{a \cov U \quad a \cov V}
    {a \cov U \downarrow V}\,
    (\downset\text{-right})\\[.5em]
    \frac{a \pos U}
    {a \elm U}\,
    (\text{coreflexivity})
    \quad
    \frac{a \pos U \quad
    \left( \forall b \in S \right) \left(  b \pos U \imp b \elm V\right)} {a \pos V}\,
    (\text{cotransitivity})\\
    \frac{a \cov U \quad  a \pos V}
    {U \pos V}\,
    (\text{compatibility})
  \end{gather*}
for all $a \in S$ and $U,V \subseteq S$, where
  \begin{align*}
    U \downarrow V &\defeql \bigcup_{a \elm U, b \elm V} a \downarrow b \\
    a \downarrow b &\defeql 
    \left\{ c \in S \mid  c \cov \left\{a \right\} \amp c \cov \left\{ b\right\} \right\} \\
    U \cov V &\defeqiv \left( \forall a \elm U \right) a \cov V \\
    U \pos V &\defeqiv \left( \exists a \elm U \right) a \pos V. \footnotemark   
  \end{align*}%
  \footnotetext{Note that $U \cov V$ is universal in character and 
    $U \pos V$ existential, which might be confusing.
    But these notations are quite useful when we define
    the notion of formal maps and 
    their images (cf.\ Definition~\ref{def:FormalMap} and 
    \eqref{eq:Image}).
    }
  A relation $\cov$ which satisfies (reflexivity), (transitivity), and
  ($\downset$-right) is called a \emph{cover} on the set $S$, and 
  a relation $\pos$ which satisfies (coreflexivity) and (cotransitivity)
  is called a \emph{positivity} on $S$. Thus, a positive topology
  is a set equipped with a compatible pair of cover and positivity.
\end{definition}

\begin{remark}
From Section~\ref{sec:ContBaire} on, we
will deal only with positive topologies in which the cover is generated by induction from some axioms.
In this case,  a positivity compatible with the cover can be generated by coinduction from the same axioms,
and it becomes the greatest positivity compatible with the
cover.
However, this does not mean that the notion of positivity is redundant because
\begin{enumerate}
  \item a positivity need not always be the greatest compatible one with
    the given cover;
  \item when no information on a cover is available about its generation, 
  one can only define the greatest positivity compatible with the cover
  through an impredicative definition (cf.\ \eqref{imprdefJ}).
\end{enumerate}
\end{remark}

\begin{notation}
  We often use letters $\mathcal{S}$ and $\mathcal{T}$ 
  to denote positive topologies of the forms
   $(S,\cov_{\mathcal{S}},\pos_{\mathcal{S}})$
  and $(T,\cov_{\mathcal{T}},\pos_{\mathcal{T}})$ respectively. 
  The subscripts attached to $\cov$ and $\pos$ are
  omitted when they are clear from the context.
\end{notation}
In the same way the definition of positive topology is obtained from a
concrete space, the definition of morphism between positive topologies
is obtained by abstraction of the properties of the right side of
a convergent relation pair. 
\begin{definition}\label{def:FormalMap}
  Let $\mathcal{S}$ and $\mathcal{T}$ be positive topologies.  A
  \emph{formal map} from $\mathcal{S}$
  to $\mathcal{T}$ is a
  relation $s \subseteq S \times T$ such that
  \begin{enumerate}[(FM1)]
  \item\label{def:FormalMap1} $S \cov_{\mathcal{S}} s^{-}T$

  \item\label{def:FormalMap2} $s^{-}b \downarrow s^{-}c \cov_{\mathcal{S}} s^{-}(b
    \downarrow c)$

  \item\label{def:FormalMap3} $b \cov_{\mathcal{T}} V \imp s^{-}b
    \cov_{\mathcal{S}} s^{-}V$

  \item\label{def:FormalMap4} $s^{-}b \pos_{\mathcal{S}} s^{*}V \imp b
    \pos_{\mathcal{T}} V$
\end{enumerate}
for all $b, c \in T$ and $V \subseteq T$,
where 
  $s^{*}V \defeql \left\{ a \in S \mid s\left\{ a \right\} \subseteq V \right\}$.\footnote{
  The operation $s^{*}$ is universal in character.
  Thus, contrary to the image operation $\Image{s}$, it cannot be defined
  as a union of its values on singletons.
  }
Two formal maps $s,s' \colon \mathcal{S} \to \mathcal{T}$ are defined to be \emph{equal} if
\begin{equation} \label{eq:EqualityPTop}
  a \cov_{\mathcal{S}} s^{-} b \leftrightarrow a
  \cov_{\mathcal{S}} s'^{-} b
\end{equation}
for all $a \in S$ and $b \in T$.
\end{definition}
The collections of positive topologies and formal maps with the
relational composition form a category $\PTop$. The identity on a
positive topology is the identity relation on its base.

Each concrete space $\mathcal{X} = (X, \Vdash, S)$ determines
a positive topology $\mathcal{S}_{\mathcal{X}} = (S,
\cov_{\mathcal{X}}, \pos_{\mathcal{X}})$ 
as follows:
\begin{align}
  \begin{aligned}
  \label{eq:Representable}
  a \cov_{\mathcal{X}} U &\defeqiv
  \ext a \subseteq \ext U \\
  a \pos_{\mathcal{X}} U &\defeqiv
  \ext a \meets \rest  U
  \end{aligned}
\end{align}
where $\rest  U  \defeql \{x \in X \mid \Diamond x \subseteq U\}$ and
$\Diamond x \subseteq S$ is the set of open neighbourhoods of
a point $x \in X$:
\begin{equation}\label{def:Diamond}
  \Diamond x \defeql \left\{ a \in S \mid x \Vdash a \right\}.
\end{equation}
While the meaning of $a \cov_{\mathcal{X}} U$ is clear,
the meaning of $a \pos_{\mathcal{X}} U$ needs some explanation:
it is easy to check that
$\rest U$ is a closed subset of $\mathcal{X}$ and the
closure of a subset $D \subseteq X$ is of the form $\rest \bigcup_{x
\in D} \Diamond x$.
Thus the closed subsets consist of subsets of the form $\rest U$ for
some $U \subseteq S$.
Hence $a \pos_{\mathcal{X}} U$
means that the open subset $\ext a$ intersects with the closed
subset represented by $U \subseteq S$.

\begin{definition}
  A positive topology is said to be \emph{representable} if it is of the
  form $\mathcal{S}_{\mathcal{X}}$ for some concrete space
  $\mathcal{X}$. 
\end{definition}
Here $\mathcal{X}$ is not a part of the structure of a representable topology.
In fact, some representable topologies admit 
purely pointfree characterisations, e.g.\ Scott topologies on
algebraic posets \cite{SambinFormalTopologyandDomains}, or more
generally on continuous posets \cite{Infosys}.

If $(r,s) \colon \mathcal{X} \to \mathcal{Y}$ is a
convergent relation pair, then $s$ 
is a formal map from $\mathcal{S}_{\mathcal{X}}$ to
$\mathcal{S}_{\mathcal{Y}}$. 

\begin{theorem} \label{thm:FundThm}
  The assignment $\mathcal{X} \mapsto \mathcal{S}_{\mathcal{X}}$ and
  $(r,s) \mapsto s$ determines a functor $\Fs \colon \CSpa \to \PTop$.
  Moreover, $\Fs$ is full and faithful.
\end{theorem}
\begin{proof}
  First, it is routine to check that $\Fs$ is a functor.
  Next, $\Fs$ is full because for every formal map $s \colon
  \Fs(\mathcal X) \to \Fs(\mathcal Y)$ one can define a relation 
  \[
    x \mathrel{r_s} y   \defeqiv  \Diamond y \subseteq s\Diamond x
  \]
  from $X$ to $Y$ such that
  $(r_s, s)$ is a convergent relation pair.
  Finally, $\Fs$ is faithful because condition \eqref{eq:EqualityPTop}
  applied to $\Fs(r,s)$ and $\Fs(r',s')$  can be shown to be
  equivalent to \eqref{eq:EqualityCSpa}.
  For details, the reader is referred to the forthcoming book \cite{Sambin_BP_book}. 
\end{proof}
Theorem~\ref{thm:FundThm} says that the
notion of positive topology is a full and faithful abstraction of the
structure induced on the base of a concrete space.  In the classical
pointfree topology \cite{johnstone-82}, this corresponds to the
embedding of the category of sober topological spaces into that of
locales.

\begin{remark}
  The notion of positive topology is richer than that of formal
  topology \cite{Sambin:intuitionistic_formal_space}, being enriched
  by positivity $\pos$. This extra structure allows us to
  prove Theorem \ref{thm:FundThm}.
  Moreover, the category of formal topologies can be embedded 
  into that of positive topologies (Ciraulo and
  Sambin~\cite{CirauloGiovanniEmbedding}).
  However, not all the practical benefits of such extension of the purely
  pointfree setting have been explored. 
\end{remark}

\subsection{Continuity theorem}\label{sec:ContThm}
The notion of ideal point of a positive topology allows us to talk about
ideal elements of the corresponding space.
An ideal point is defined
abstracting the properties of the neighbourhoods $\Diamond x$ of an
element $x$ of a concrete space.
\begin{definition}
Let $\mathcal{S}$ be a positive topology. An \emph{ideal point}
is a subset $\alpha \subseteq S$ such that 
\begin{enumerate}
  \item $\alpha$ is \emph{inhabited}, i.e.
    $\alpha \meets \alpha$;

  \item $\alpha$ is \emph{filtering}, i.e.
    $a,b \elm \alpha \imp \alpha \meets (a \downarrow b)$\:  for all $a,b \in S$;

  \item $\alpha$  \emph{splits  $\cov$}, i.e.
    $a \cov U \amp  a \elm \alpha \imp \alpha \meets U$\: for all $a \in S$ and
    $U \subseteq S$;

  \item $\alpha$ \emph{enters  $\pos$}, i.e.
    $a \elm \alpha \subseteq V \imp a \pos V$\: for all $a \in S$ and $V \subseteq
    S$.
\end{enumerate}
The collection of ideal points of a positive topology
$\mathcal{S}$ is denoted by $\Pt{\mathcal{S}}$.
\end{definition}
The collection  $\Pt{\mathcal{S}}$ is equipped with a pointwise
topology generated by open subcollections of the form
\[
 \Ext(a) \defeql \left\{ \alpha \colon \Pt{\mathcal{S}} \mid a \elm \alpha \right\}
\]
for each $a \in S$. 
This topology can be represented by a \emph{large} concrete space
$\Ip(\mathcal{S}) = (\Pt{\mathcal{S}},
\backelm, S)$,
where $\alpha \backelm a \defeqiv a \elm \alpha$.

The modifier \emph{large} is due to the fact that the left-hand side
of $\Ip(\mathcal{S})$ is a collection rather than a set, and thus it is
not a concrete space in a proper sense. 
However, it is convenient to consider such structures, and we will do so in the following.
 The reason why it is safe to consider large structures is that we work on them 
 using only predicate logic, which applies to collections as well as to sets. 
 Thus, the results on concrete spaces that depend only on logic apply
 automatically to large concrete spaces.
 The same remark applies to large covers and large positivities, which
 will be introduced at the beginning of Section~\ref{sec:Cont}.
\begin{definition}\label{def:ContMap}
If $\mathcal{S}$ and $\mathcal{T}$ are positive topologies, 
a \emph{continuous map} from  $\Ip(\mathcal{S})$
to $\Ip(\mathcal{T})$ is a pair $(g,s)$ where $g \colon \Pt{\mathcal{S}}
\to \Pt{\mathcal{T}}$ is an operation and $s \subseteq S \times T$ is a
relation which makes the following diagram commute:
\[
 \begin{tikzcd}
   \Pt{\mathcal{S}}
 \arrow[r,"\backelm"] \arrow[d,"g"] &
 S \arrow[d,"s"] \\
 \Pt{\mathcal{T}} \arrow[r, "\backelm"] 
 & T
 \end{tikzcd}
\]
\end{definition}
The above definition
is motivated by the following observation on representable topologies:
let $(r,s)\colon \mathcal{X} \to \mathcal{Y}$ be a convergent relation pair.
Then $s$ is a formal map from $\mathcal{S}_{\mathcal{X}}$ to $\mathcal{S}_{\mathcal{Y}}$
and hence (as we will see in Corollary~\ref{cor:FMapImpCont} below)
$s_*$ is an operation from 
$\Pt {\mathcal{S}_{\mathcal{X}}}$ to $\Pt {\mathcal{S}_{\mathcal{Y}}}$. 
Since $(s_*,s)$ trivially makes the above square commute,
it is a continuous map from $\Pt {\mathcal{S}_{\mathcal{X}}}$ 
to $\Pt {\mathcal{S}_{\mathcal{Y}}}$.

Every continuous map $(g,s)$ is pointwise continuous, that is
\[
  b \elm g(\alpha) \imp  \exists a \elm \alpha 
  \left[ \left( \forall \beta \colon \Pt{\mathcal{S}} \right) a \elm \beta \imp b \elm g(\beta)
  \right]
\]
for all $\alpha \colon \Pt{\mathcal{S}}$ and $b \in T$. More
specifically, from
the open neighbourhood $\Ext(b)$ of $g(\alpha)$, 
by commutativity one can trace $b$ backward along $s$
and find an open neighbourhood $\Ext(a)$ of $\alpha$ whose image
under $g$ falls within $\Ext(b)$. Thus, the relation $s$ acts as a
\emph{modulus} of continuity for the operation $g$.

Commutativity of the square means that
\[
  b \elm g(\alpha)
  \leftrightarrow s^{-}b \meets \alpha \leftrightarrow b \elm
  \Image{s}(\alpha)
\]
for all $b \in T$ and $\alpha \colon \Pt{\mathcal{S}}$,
 which is equivalent to saying that $g$ is equal to the image
 operation $\Image{s} \colon \Pow{S} \to \Pow{T}$ on
 $\Pt{\mathcal{S}}$. Hence, every continuous map
 from $\Ip(\mathcal{S})$  to $\Ip(\mathcal{T})$ is of the form
 $(\Image{s},s)$ where $s \subseteq S \times T$ is a relation which
 maps $\alpha \colon \Pt{\mathcal{S}}$ to $\Image{s}(\alpha) \colon
 \Pt{\mathcal{T}}$.  Conversely, if $s \subseteq S \times T$ is a
 relation which induces a well-defined operation $\Image{s} \colon
 \Pt{\mathcal{S}} \to \Pt{\mathcal{T}}$, then the pair $(\Image{s},s)$
 clearly makes the square commute.
 In summary, we have the following theorem.
 \begin{theorem}[Continuity theorem]
  \label{thm:PWCont}
  Let $\mathcal{S}$ and $\mathcal{T}$ be positive topologies, and 
  let $s \subseteq S \times T$ be a relation. 
  If $\Image{s}$ is a
  mapping from
  $\Pt{\mathcal{S}}$ to $\Pt{\mathcal{T}}$, then $(\Image{s},s)$ is a continuous map from
  $\Ip(\mathcal{S})$ to $\Ip(\mathcal{T})$. Moreover, every
  continuous map from $\Ip(\mathcal{S})$ to $\Ip(\mathcal{T})$ is
  induced by a relation from $S$ to $T$ in this way.
\end{theorem}
Theorem \ref{thm:PWCont} should be compared to the familiar form of
continuity principle ``every full function is continuous''.  The
theorem articulates a condition on operations between ideal points in
which the continuity principle holds, i.e.\ that of being induced by a
relation between bases.  This condition is reasonable from the
constructive point of view since in order to define an operation on an
infinite object like an ideal point, we can
only rely on finite information about it, i.e.\ its formal
neighbourhoods and a relation between them (cf.\ Vickers
\cite{vickers1989topology}). From a predicative point of
view, we believe that the notion of continuous map is one of the
simplest way of characterising continuous operations between
ideal points.

\subsection{Principle of pointfree continuity}\label{sec:Cont}
Our next aim is to relate the notion of continuous map to that of
formal map.

Given a positive topology $\mathcal{S}$, let
$\mathcal{S}_{\Ip} = (S, \cov_{\Ip}, \pos_{\Ip})$ be the \emph{large}
positive topology associated with the concrete space
$\Ip(\mathcal{S})$ (cf.\ \eqref{eq:Representable}).
Note that $\cov_{\Ip}$ and $\pos_{\Ip}$ are defined by
quantifications over $\Pt{\mathcal{S}}$, which is not
necessarily a set, and hence they are \emph{large} structures in general.
There is a formal map $\varepsilon_{\mathcal{S}} \colon
\mathcal{S}_{\Ip} \to \mathcal{S}$ represented by the identity
relation on $S$.  In particular, we have
\begin{align*}
&a \cov U \imp a \cov_{\Ip} U
&a \pos_{\Ip} U \imp a \pos U
\end{align*}
for all $a \in S$ and $U \subseteq S$.

\begin{lemma}\label{lem:FundLemma}
Let $s \subseteq S \times T$ be a relation between the bases
of positive topologies $\mathcal{S}$ and $\mathcal{T}$. Then
\begin{enumerate}
  \item\label{lem:FundLemma1} $S \cov_{\mathcal{S}_{\Ip}} s^{-}T$ if and only if
    $\Image{s}(\alpha)$ is inhabited for all $\alpha \colon
    \Pt{\mathcal{S}}$;

  \item\label{lem:FundLemma2} $s^{-}b \downarrow s^{-}c
    \cov_{\mathcal{S}_{\Ip}}
    s^{-}(b \downarrow c)$ for all $b,c \in T$ if and only if
    $\Image{s}(\alpha)$ is filtering for all $\alpha \colon
    \Pt{\mathcal{S}}$;

  \item\label{lem:FundLemma3} $b \cov_{\mathcal{T}}V \imp s^{-}b
    \cov_{\mathcal{S}_{\Ip}} s^{-}V$ for all $b \in T$ and $V \subseteq T$ if and
    only if $\Image{s}(\alpha)$ splits $\cov_{\mathcal{T}}$ for all
    $\alpha \colon \Pt{\mathcal{S}}$;

  \item\label{lem:FundLemma4} $s^{-}b \pos_{\mathcal{S}_{\Ip}}s^{*}V
    \imp b \pos_{\mathcal{T}} V$ for all
    $b \in T$ and $V \subseteq T$ if and only if $\Image{s}(\alpha)$ enters
    $\pos_{\mathcal{T}}$ for all $\alpha \colon \Pt{\mathcal{S}}$.
\end{enumerate}
\end{lemma}
\begin{proof}
  \noindent \eqref{lem:FundLemma1}
   Since every ideal point is inhabited, $S \cov_{\mathcal{S}_{\Ip}} s^{-}T$
   is equivalent to $\alpha \meets s^{-}T$ for all $\alpha \colon
   \Pt{\mathcal{S}}$, that is $\Image{s}(\alpha) \meets  T$ for
   all $\alpha \colon \Pt{\mathcal{S}}$.
   \smallskip

  \noindent\eqref{lem:FundLemma2} Similar to  \eqref{lem:FundLemma1} using
  the fact that every ideal point is filtering.
   \smallskip

  \noindent \eqref{lem:FundLemma3}
  By unfolding the definition,
  $b \cov_{\mathcal{T}}V \imp s^{-}b
    \cov_{\mathcal{S}_{\Ip}} s^{-}V$ for all $b \in T$ and $V
    \subseteq T$
    if and only if
   $b \cov_{\mathcal{T}}V \imp  \forall \alpha \colon 
   \Pt{\mathcal{S}} \left[  \alpha \meets s^{-}b
     \imp
    \alpha \meets s^{-}V \right]$ for all $b \in T$ and $V \subseteq T$. This is easily seen to be equivalent to 
   $b \cov_{\mathcal{T}}V \amp  b \elm \Image{s}(\alpha) \imp
   \Image{s}(\alpha) \meets V $ for all $b \in T$, $V \subseteq T$
    and $\alpha \colon \Pt{\mathcal{S}}$, that is
   $\Image{s}(\alpha)$ splits $\cov_{\mathcal{T}}$ for all $\alpha
   \colon \Pt{\mathcal{S}}$.
   \smallskip

  \noindent\eqref{lem:FundLemma4} By unfolding the definition, we have
  $s^{-}b \pos_{\mathcal{S}_{\Ip}} s^{*}V \imp b \pos_{\mathcal{T}} V$
  for all $b \in T$ and $V \subseteq T$
  if and only if
  $ \exists \alpha \colon \Pt{\mathcal{S}}  \left[ \alpha
  \meets s^{-}b  \amp \alpha \subseteq s^{*}V \right]  \imp b
  \pos_{\mathcal{T}} V$
  for all $b \in T$ and $V \subseteq T$. This is equivalent to
  $b \elm \Image{s}(\alpha) \amp \Image{s}(\alpha) \subseteq V  
    \imp b \pos_{\mathcal{T}} V$
    for all $b \in T$, $V \subseteq T$ and $\alpha \colon
    \Pt{\mathcal{S}}$, that is
    $\Image{s}(\alpha)$  enters $\pos_{\mathcal{T}}$ for all
    $\alpha \colon \Pt{\mathcal{S}}$.
\end{proof}

\begin{proposition}\label{prop:EquiContOp}
  Let $\mathcal{S}$ and $\mathcal{T}$ be positive topologies
  and $s \subseteq  S \times T$ be a relation. The following are
  equivalent:
  \begin{enumerate}
  \item\label{lam:EquiContOp1} $\Image{s}(\alpha) \colon
    \Pt{\mathcal{T}}$ for all $\alpha \colon \Pt{\mathcal{S}}$;
  \item\label{lam:EquiContOp2} $(\Image{s}, s)$ is a continuous
    map from $\Ip(\mathcal{S})$ to
    $\Ip(\mathcal{T})$;
  \item\label{lam:EquiContOp4} 
    $s$ is a formal map from $\mathcal{S}_{\Ip}$ to $\mathcal{T}$.
  \end{enumerate}
\end{proposition}
\begin{proof}
  (\ref{lam:EquiContOp1} $\Leftrightarrow$ \ref{lam:EquiContOp2})
  This is the content of Theorem \ref{thm:PWCont}.

  \noindent(\ref{lam:EquiContOp1} $\Leftrightarrow$ \ref{lam:EquiContOp4})
  Immediate from Lemma \ref{lem:FundLemma}.
\end{proof}

\begin{corollary}\label{cor:FMapImpCont}
  If $s \colon \mathcal{S} \to \mathcal{T}$ is a formal map,
  then $(\Image{s},s)$ is continuous, i.e.\ $\Image{s}$ is a mapping from
  $\Pt{\mathcal{S}}$ to $\Pt{\mathcal{T}}$.
\end{corollary}
\begin{proof}
  If $s \colon \mathcal{S} \to \mathcal{T}$ is a formal map,
  then composition with the canonical map $\varepsilon_{\mathcal{S}}
  \colon \mathcal{S}_{\Ip} \to \mathcal{S}$ gives a formal map $s
  \colon \mathcal{S}_{\Ip} \to \mathcal{T}$. Then, the conclusion
  follows from Proposition \ref{prop:EquiContOp}.
\end{proof}

If the pointwise cover and positivity coincide with the pointfree 
ones, i.e.\ the canonical map $\varepsilon_{\mathcal{S}} \colon
\mathcal{S}_{\Ip} \to \mathcal{S}$ is an isomorphism, then we could
have replaced item \eqref{lam:EquiContOp4} of Proposition
\ref{prop:EquiContOp} with a formal map $s \colon \mathcal{S} \to
\mathcal{T}$. Here, the notion of bi-spatiality is exactly what is
required.
\begin{definition}\label{def:bi-spatialty}
  A positive topology $\mathcal{S}$ is $\emph{bi-spatial}$ if
  the canonical formal map $\varepsilon_{\mathcal{S}} \colon
  \mathcal{S}_{\Ip} \to \mathcal{S}$ is an isomorphism, i.e.
  \begin{align*}
  &a \cov_{\Ip} U \imp  a \cov U && \text{(spatiality)}\\
  &a \pos U  \imp   a \pos_{\Ip} U && \text{(reducibility)} 
  \end{align*}
   for all $a \in S$ and $U \subseteq S$.
   A positive topology is \emph{spatial} if 
   it satisfies spatiality.
\end{definition}
The following proposition corresponds to the embedding
of the category of sober topological spaces into that of spatial
locales (Johnstone \cite[Chapter II, 1.7]{johnstone-82}).
\begin{proposition}\label{prop:RepBiSpatial}
  Every representable topology is bi-spatial.
\end{proposition}
\begin{proof}
  This follows from the fact that in any concrete space
  $\mathcal{X}$, the subset $\Diamond x$ is an ideal point of
  $\mathcal{S}_{\mathcal{X}}$ for each $x \in X$.
\end{proof}

From Definition \ref{def:bi-spatialty} and  Proposition 
\ref{thm:PWCont}, we obtain our second continuity theorem:
\begin{theorem}\label{thm:biSpatImpCont}
  Let $\mathcal{S}$ and $\mathcal{T}$ be positive topologies, and $s
  \subseteq S \times T$ be a relation.  If $\mathcal{S}$ is bi-spatial, then
  $s$ is a formal map $s \colon \mathcal{S} \to \mathcal{T}$ if and only if
  $\Image{s}(\alpha) \colon \Pt{\mathcal{T}}$ for all $\alpha \colon
  \Pt{\mathcal{S}}$.
\end{theorem}

Our new principle of continuity is then obtained by omitting the
assumption of bi-spatiality from Theorem \ref{thm:biSpatImpCont}.
\begin{definition}\label{def:Cont}
Let $\mathcal{S}$ and $\mathcal{T}$ be positive topologies.
The \emph{principle of pointfree continuity for $\mathcal{S}$ and
$\mathcal{T}$} is the statement:
\begin{description}
  \item[($\Cont_{\mathcal{S}, \mathcal{T}}$)]
For any relation  $s \subseteq S \times T$, if
$\Image{s}$ is a mapping from $\Pt{\mathcal{S}}$ to $\Pt{\mathcal{T}}$,
then $s$ is a formal map $s \colon \mathcal{S} \to \mathcal{T}$.
\end{description}
\end{definition}
After this definition, 
the content of Theorem~\ref{thm:biSpatImpCont} may be expressed by:  
$\Cont_{\mathcal{S}, \mathcal{T}}$ holds whenever $\mathcal{S}$ is bi-spatial.
\begin{remark}
  $\Cont_{\mathcal{S},\mathcal{T}}$ is concerned with the property of a relation $s
\subseteq  S \times T$, and not with the property of the operation
$\Image{s} \colon \Pt{\mathcal{S}} \to \Pt{\mathcal{T}}$. Another
possible formulation of the continuity principle would be to say that
every continuous operation $(g, s) \colon \Ip(\mathcal{S}) \to \Ip(\mathcal{T})$
is induced by a formal map $s' \colon \mathcal{S} \to
\mathcal{T}$, namely $g = \Image{s'}$, or equivalently
$\Image{s} = \Image{s'}$. This formulation of continuity principle is
analogous to the uniform continuity theorem for the Cantor space, saying
that every pointwise continuous function from the Cantor space to
the discrete space of natural numbers is uniformly continuous. This latter form is
weaker than the one given in Definition \ref{def:Cont} (cf.\ Theorem
\ref{thm:PiBiEquivInducedFMap}).
\end{remark}

With no restriction on $\mathcal{S}$ and $\mathcal{T}$,
the principle $\Cont_{\mathcal{S},\mathcal{T}}$ is false.
The following examples are not surprising, given the fact
that $\Cont_{\mathcal{S},\mathcal{T}}$ is a kind of completeness principle.
\begin{example}[Non-spatial topologies] \leavevmode
  \begin{enumerate}
    \item Consider a positive topology $\mathcal{S}$ with
      $S = \left\{ * \right\}$, and
    $a \cov U \defeqiv a \elm U$ and $a \pos U \defeqiv \bot$
    for all $a \in S$ and $U \subseteq S$.
    The topology $\mathcal{S}$ has no points,
    i.e.\ $\Pt{\mathcal{S}} = \emptyset$.
    For any positive topology $\mathcal{T}$, if the principle
    $\Cont_{\mathcal{S},\mathcal{T}}$ holds, then any relation $s
    \subseteq S \times T$ is a formal map from $\mathcal{S}$ to
    $\mathcal{T}$. But the empty relation cannot be a formal map since
    it does not satisfy \ref{def:FormalMap1}.

    \item There is a counterexample which is more natural than the
      previous one, the geometric theory of surjective functions
      between sets (cf.\ Fox \cite[Section 4.1.4]{Fox05}).  This
      topology, call it  $\mathcal{S}$, is a positive topology with no
      points, but there exists $a \in S$ such that $a \pos S$.
      Then, the same argument as in the first example leads to a
      contradiction.
  \end{enumerate}
\end{example}

In the light of the counterexamples above, we give some sufficient
conditions on
$\mathcal{S}$ and $\mathcal{T}$ under which
$\Cont_{\mathcal{S}, \mathcal{T}}$ holds. Note that by Theorem
\ref{thm:biSpatImpCont} if $\mathcal{S}$ is bi-spatial
then $\Cont_{\mathcal{S}, \mathcal{T}}$ holds.
For example,
$\Cont_{\mathcal{S},\mathcal{T}}$ holds whenever
$\mathcal{S}$ is representable (cf.\ Proposition
\ref{prop:RepBiSpatial}).

Let $s \subseteq S \times T$ be a relation between the underlying
bases of positive topologies $\mathcal{S}$ and $\mathcal{T}$.
Define relations $\cov_{s}$ and $\pos_{s}$ from $T$ to $\Pow{T}$ as
\begin{align}
  b \cov_{s} V &\defeqiv s^{-}b \cov_{\mathcal{S}} s^{-}V &
  b \pos_{s} V &\defeqiv s^{-}b \pos_{\mathcal{S}} s^{*}V.
  \label{eq:Image}
\end{align}
We write $\im[s]$ for the structure $(T, \cov_{s},  \pos_{s})$.
\begin{lemma}
  \label{lem:Image}
  \leavevmode
  \begin{enumerate}
    \item\label{lem:Image1} $\im[s]$ satisfies all the properties of
      positive topology except ($\downset$-right).

    \item\label{lem:Image2} If $s$ satisfies \ref{def:FormalMap2} and 
      \ref{def:FormalMap3}, then $\im[s]$ is a positive topology.

    \item\label{lem:Image3} If $s$ satisfies \ref{def:FormalMap2}
      and \ref{def:FormalMap3}, then the relation $s$ is a
      formal map from $\mathcal{S}$ to $\im[s]$ if and only if
      $s$ satisfies \ref{def:FormalMap1}.

    \item\label{lem:Image4} 
      If $s$ satisfies \ref{def:FormalMap2} and \ref{def:FormalMap3},
      then the identity relation $\id_{T}$ on $T$ is a formal map from
      $\im[s]$ to $\mathcal{T}$ if and only if $s$ satisfies
      \ref{def:FormalMap4}.
  \end{enumerate}
\end{lemma}
\begin{proof}
\noindent\eqref{lem:Image1}
It is straightforward to show that $\im[s]$
satisfies
all the properties of positive topology except ($\downset$-right)
using the corresponding properties of $\mathcal{S}$.
For example, to see that $\cov_{s}$ and $\pos_{s}$ satisfy
(compatibility), suppose we have $b \cov_{s} V$ and 
$b \pos_{s} W$. Then there exists $a \elm s^{-} b$ such that $a
\pos_{\mathcal{S}} s^{*}W$. Since $a \cov_{\mathcal{S}} s^{-}V$,
we have $ s^{-}V \pos_{\mathcal{S}} s^{*}W$ by
(compatibility) of $\mathcal{S}$. Hence $V \pos_{s} W$.
\smallskip

\noindent\eqref{lem:Image2}
Assume that $s$ satisfies \ref{def:FormalMap2} and \ref{def:FormalMap3}.
By \eqref{lem:Image1},
it suffices to show that  $\im[s]$ satisfies ($\downset$-right).
Suppose that $b \cov_{s} V$ and $b \cov_{s} W$. By ($\downset$-right) 
of $\mathcal{S}$, we have
$s^{-}b \cov_{\mathcal{S}} s^{-}V \downarrow_{\mathcal{S}} s^{-}W$.
Then $s^{-}V \downarrow_{\mathcal{S}} s^{-}W \cov_{\mathcal{S}} s^{-}(V
\downarrow_{\mathcal{T}} W)$ by
\ref{def:FormalMap2}. Hence by \ref{def:FormalMap3} and
(transitivity) of $\cov_{s}$, we have  $b \cov_{s} V
\downarrow_{\im[s]} W$.
   \smallskip

\noindent\eqref{lem:Image3}
Assume that $s$ satisfies \ref{def:FormalMap2} and \ref{def:FormalMap3}.
It is easy to see that  $s$ (as a formal map from $\mathcal{S}$ to
$\im[s]$) satisfies \ref{def:FormalMap2},  \ref{def:FormalMap3}, and
\ref{def:FormalMap4}. Then condition \ref{def:FormalMap1} is equivalent
to that of $s$ (as a formal map from $\mathcal{S}$ to $\mathcal{T}$).
\smallskip

\noindent\eqref{lem:Image4}
Assume that $s$ satisfies \ref{def:FormalMap2} and \ref{def:FormalMap3}.
Then $\id_{T}$ trivially satisfies \ref{def:FormalMap1}.
For \ref{def:FormalMap2}, let $d \elm {\id_{T}}^{-} b
\downarrow_{\im[s]} {\id_{T}}^{-} c$. Then $s^{-} d
\cov_{\mathcal{S}}
s^{-} b$ and $s^{-} d \cov_{\mathcal{S}} s^{-} c$. By 
\ref{def:FormalMap2} of $s$, we have $s^{-}d \cov_{\mathcal{S}}
s^{-}(b \downarrow_{\mathcal{T}} c)$, i.e.\ $d \cov_{s}
{\id_{T}}^{-}(b \downarrow_{\mathcal{T}} c)$. Therefore
${\id_{T}}^{-} b \downarrow_{\im[s]} {\id_{T}}^{-} c
\cov_{s}
{\id_{T}}^{-}(b \downarrow_{\mathcal{T}} c)$.
Lastly, conditions \ref{def:FormalMap3} and
\ref{def:FormalMap4} for $\id_{T}$ is equivalent to the corresponding
conditions of $s$.
\end{proof}

\begin{definition}\label{def:Image}
  The \emph{image} of a formal map $s \colon \mathcal{S} \to
  \mathcal{T}$ is the positive topology $\im[s] = (T, \cov_{s},
  \pos_{s})$.
\end{definition}
As a corollary of Lemma \ref{lem:Image}, we have an image
factorisation.
\begin{proposition}\label{prop:ImageFact}
  Any formal map $s \colon \mathcal{S} \to \mathcal{T}$ between
  positive topologies factors as $s \colon \mathcal{S} \to
  \im[s]$ and $\id_{T} \colon \im[s] \to \mathcal{T}$.
\end{proposition}

\begin{proposition}\label{prop:ImagePoPC}
  For any formal map $s \colon \mathcal{S} \to \mathcal{T}$
  and a positive topology $\mathcal{S}'$, the principle
  $\Cont_{\mathcal{S},\mathcal{S}'}$ implies
  $\Cont_{\im[s],\mathcal{S}'}$.
\end{proposition}
\begin{proof}
  Fix a formal map $s \colon \mathcal{S} \to \mathcal{T}$ and
  a positive topology $\mathcal{S}'$.
  It is easy to see that a relation $s' \subseteq T
  \times S'$ is a formal map from $\im[s]$ to $\mathcal{S}'$ if
  and only if the composition $s' \circ s$ is a formal map from
  $\mathcal{S}$ to $\mathcal{S}'$. 
  
  Let $s' \subseteq T \times
  S'$ be a relation such that $\Image{s'}$ is a mapping from
  $\Pt{\im[s]}$ to $\Pt{\mathcal{S}'}$.
  Since $\Image{s}$ is a mapping from $\Pt{\mathcal{S}}$ to
  $\Pt{\im[s]}$, we have that
  $\Image{(s' \circ s)}$ is a mapping from $\Pt{\mathcal{S}}$ to
  $\Pt{\mathcal{S}'}$. Thus $s$ is a formal map from
  $\mathcal{S}$ to $\mathcal{S}'$ by
  $\Cont_{\mathcal{S},\mathcal{S}'}$. Therefore $s'$ is a formal map
  from $\im[s]$ to $\mathcal{S}'$.
\end{proof}

Given a cover $\cov$ on a set $S$, we can impredicatively define a
positivity $\pos_{\cov}$ which is compatible with $\cov$ as follows:
\begin{align} 
\label{imprdefJ} 
  a \pos_{\cov} U 
  \defeqiv
  \left( \exists V \colon \Pow{S} \right)
  \left[  
  a \elm V \subseteq U \amp
  \left( \forall W \colon \Pow{S} \right) V \meets \sat W \imp V \meets W
   \right]
\end{align}
where
  $
  \sat W \defeql \left\{ a \in S \mid a \cov W \right\}.
  $
  It is easy to see that $\pos_{\cov}$ is the 
  greatest positivity
compatible with $\cov$.
If a cover $\cov$ is generated inductively (cf.\ Coquand et al.\
\cite{Coquand200371}), then the positivity $\pos_{\cov}$ can be constructed by
coinduction \cite{Sambin_BP_book}, and hence its construction can be
done predicatively. This is the case for the positive topologies
treated in Section~\ref{sec:ContBaire} and Section~\ref{sec:ContReal}, whose
covers are generated inductively.

\begin{lemma}
  \label{lem:Compatibility}
  Let $\cov$ and $\cov'$ be two covers on a set $S$, and let
  $\pos$ be a positivity on $S$. If $\cov' \subseteq \cov$ and $\pos$ is
  compatible with $\cov$, then $\pos$ is compatible with $\cov'$.
\end{lemma}
\begin{proof}
  Obvious from the definition of compatibility.
\end{proof}

\begin{corollary}\label{col:MaxCoDomFormalMap}
  Let $\mathcal{S}$ and $\mathcal{T}$ be positive topologies,
  and  $s \subseteq S \times T$ be a relation which satisfies 
  \ref{def:FormalMap1}, \ref{def:FormalMap2}, and \ref{def:FormalMap3}.
  If $\pos_{\mathcal{T}}$ is the 
  greatest positivity compatible with
  $\cov_{\mathcal{T}}$, then $s$ is a formal map from $\mathcal{S}$ to
  $\mathcal{T}$.
\end{corollary}
\begin{proof}
  Condition \ref{def:FormalMap4}
  is equivalent to saying that
  the positivity $\pos_s$ is smaller than
  $\pos_{\mathcal{T}}$. By \ref{def:FormalMap3}, we have
  $\cov_{\mathcal{T}} \subseteq \cov_{s}$. Thus
  $\pos_{s}$ is compatible with $\cov_{\mathcal{T}}$ by Lemma
  \ref{lem:Compatibility}. Therefore $\pos_{s} \subseteq \pos_{\mathcal{T}}$.
\end{proof}

\begin{proposition}\label{prop:SpatImpCont}
  Let $\mathcal{S}$ and $\mathcal{T}$ be positive topologies.
  If $\mathcal{S}$ is spatial and
  $\mathcal{T}$ has the greatest
  positivity compatible with its cover,
  then $\Cont_{\mathcal{S},\mathcal{T}}$ holds.
\end{proposition}
\begin{proof}
  Immediate from Corollary \ref{col:MaxCoDomFormalMap} and 
  Proposition \ref{prop:EquiContOp}.
\end{proof}
Most of the positive topologies which arise in practice have
greatest
compatible positivities, so the assumption on the topology
$\mathcal{T}$ in Proposition \ref{prop:SpatImpCont} is often
satisfied. Thus, it is the spatiality of $\mathcal{S}$
that is crucial for the continuity principle
$\Cont_{\mathcal{S},\mathcal{T}}$ to hold.

At this point, it is natural to ask whether (bi-)spatiality of
$\mathcal{S}$ is actually necessary for the principle
$\Cont_{\mathcal{S},\mathcal{T}}$ to hold. In the following sections, we
answer this question in some specific cases which occupy a central
place in Brouwer's intuitionism.

\section{Continuity on the Baire space}\label{sec:ContBaire}
The continuity principle for the Baire space deserves special attention.
Since this is the context in which Brouwer introduced his principle of
continuity, it is of our particular interest to see in what sense our
continuity principle is related to the principles of intuitionism.

\subsection{Formal Baire space}\label{sec:Baire}
We recall the pointfree definition of the Baire space,
whose ideal points can
be considered as free choice sequences.
  Let $\FSeq$ denote 
   the set of finite sequences of natural numbers.
  We write $a * b$ for the concatenation of finite sequences $a$
  and $b$. By an abuse of notation, we write $a * n$ for the
  concatenation of a finite sequence $a$ and the singleton sequence of
  $n \in \Nat$.
  The order $\leq_{\FBaire}$ on $\FSeq$ is defined by the reverse prefix
  ordering:
  \[
    a \leq_{\FBaire} b \defeqiv \text{$b$ is an initial segment of $a$.}
  \]
  In particular, we have $a * n \leq_{\FBaire} a$.

\begin{definition}  The \emph{formal Baire space} is a positive topology  $\FBaire = (\FSeq,
  \cov_{\FBaire}, \pos_{\FBaire})$ where the cover $\cov_{\FBaire}$
  is inductively generated by the following rules
  \begin{gather*}
    \frac{a  \elm U}{a \cov_{\FBaire} U}\,(\eta) \qquad
    \frac{a \leq_{\FBaire} b \cov_{\FBaire}U}{a \cov_{\FBaire}
  U}\,(\zeta) \qquad
  \frac{\left( \forall n \in \Nat \right) a * n \cov_{\FBaire}
U}{a \cov_{\FBaire} U}\,(\digamma)
\end{gather*}
and $\pos_{\FBaire}$ is the greatest positivity compatible with
$\cov_{\FBaire}$. 
\end{definition}

We recall some well-known properties of $\FBaire$ that
we shall use in this section.
Let $\downset_{\FBaire}{U}$ denote the
downward closure of a subset $U \subseteq \FSeq$ with respect to
$\leq_{\FBaire}$:
\begin{equation*}
  \downset_{\FBaire} {U} \defeql \left\{ a \in \FSeq \mid \left(
    \exists b \elm U \right) a \leq_{\FBaire} b \right\}.
  \end{equation*}
The $\zeta$-inference can be eliminated in the
following sense; see e.g.\ Troelstra and {van Dalen} \cite[Chapter 4, Exercise 4.8.10]{ConstMathI}.
\begin{lemma} \label{lem:ElimZeta}
  Let $\bcov_{\FBaire}$ be the relation from $\FSeq$ to $\Pow{\FSeq}$
  inductively defined by $\eta$ and $\digamma$-rules. Then
  \[
    a \cov_{\FBaire} U \leftrightarrow a \bcov_{\FBaire}
    \downset_{\FBaire}{U}
  \]
  for all $a \in \FSeq$ and $U \subseteq \FSeq$.
\end{lemma}
\begin{proof}
  By induction on $\cov_{\FBaire}$ and $\bcov_{\FBaire}$.
\end{proof}
Using Lemma \ref{lem:ElimZeta}, it is straightforward to show
\begin{equation}\label{eq:ConvOrd}
  a \leq_{\FBaire} b  \leftrightarrow a \cov_{\mathcal{B}} \left\{ b \right\}
\end{equation}
for all $a,b \in \FSeq$. From Lemma \ref{lem:ElimZeta} and
\eqref{eq:ConvOrd}, we can obtain the following elementary characterisation of
an ideal point of $\FBaire$.
\begin{proposition}\label{col:CSasIPt}
  A subset $\alpha \subseteq \FSeq$ is an ideal point of $\FBaire$ if and only if
\begin{enumerate}
  \item $\alpha \meets \alpha$;
  \item $a, b \elm \alpha \imp 
    \alpha \meets (a \downarrow b)$ for all
    $a,b \in \FSeq$;
  \item $\cons{a}{x} \elm \alpha \imp  a \elm \alpha$ for all
    $a \in \FSeq$ and $x \in \Nat$;
  \item\label{col:CSasIPt4} $a \elm \alpha \imp  \left( \exists x \in \Nat \right)
    \cons{a}{x} \elm \alpha$ for all
    $a \in \FSeq$.
\end{enumerate}
\end{proposition}

The above characterisation
provides us with a geometric intuition of an ideal point of
$\FBaire$ as an infinitely proceeding sequence 
\[
  \nil, x_{0},  x_{0} * x_{1},  x_{0} * x_{1} * x_{2}, \cdots
\]
of one-step extensions of elements of $\FSeq$ starting from the empty
sequence $\nil$; at each
``stage'' $x_{0} * \cdots * x_{n}$ of this sequence, the next step
$x_{0} * \cdots * x_{n} * x$ is constructed by choosing an arbitrary
element $x \in \Nat$ without any restriction (cf.\ condition
\eqref{col:CSasIPt4} of Proposition \ref{col:CSasIPt}).
Thus, the notion of ideal point of $\FBaire$ can be considered as
a possible manifestation of free choice sequences \cite[Chapter 3]{DummettElemInt}. This led us
to identify choice sequences with ideal points of
$\FBaire$.
\begin{definition}
A \emph{choice sequence} is an ideal point of $\FBaire$.
\end{definition}

We here leave the reader to check that in the Minimalist Foundation 
(in particular, without assuming any form of axiom of choice) one can prove that:
\begin{theorem}
There is a bijective correspondence between choice sequences, i.e.\ ideal points of $\FBaire$,
and functions from 
$\Nat$ to $\Nat$, i.e.\ total and single-valued relations from $\Nat$
to $\Nat$.
\end{theorem}

\subsection{Continuity principles for the Baire space}\label{sec:ContBN}
The principle of bar induction and the principle of continuity for functions from
choice sequences to the natural numbers are closely related.
In particular, Brouwer introduced bar induction in his analysis of the
computation tree of such a function.  Hence, we focus on the
continuity principle
between the Baire space and the natural numbers in the pointfree
setting, and study its connection to various forms of bar induction.
The standard reference of bar induction is Troelstra and van Dalen
\cite[Chapter 4, Section 8]{ConstMathI}.

\begin{definition}
For every set $S$, we define the \emph{discrete positive topology} 
on $S$ to be a positive topology
$\Discrete{S} = (S, \cov_{\Discrete{S}}, \pos_{\Discrete{S}})$
where cover and positivity are defined by
\begin{align*}
  a \cov_{\Discrete{S}} U &\defeqiv a \elm U, &
  a \pos_{\Discrete{S}} U &\defeqiv a \elm U.
\end{align*}
It is easy to see that $\pos_{\Discrete{S}}$ is the 
greatest positivity compatible with
$\cov_{\Discrete{S}}$, and that ideal points of $\Discrete{S}$ are
singletons of $S$.
The \emph{formal topology of natural numbers} is the discrete positive topology
$\Discrete{\Nat}$ on the set $\Nat$ of natural numbers,
which we denote by
$\FNat = \left(\Nat, \cov_{\FNat}, \pos_{\FNat} \right)$.
\end{definition}
A relation $s \subseteq X \times Y$ from a poset $(X,\leq)$ to a set
$Y$ is said to be \emph{monotone} if 
\[
  \downarrow s^{-} y = s^{-}y
\]
for all $y \in Y$. The \emph{monotonisation} of a relation $s
\subseteq X \times Y$ is the composition $s \mathop{\circ} \leq$, which
is obviously monotone.
\begin{lemma}\label{lem:Monotonisation}
Every formal map $s
\colon \FBaire \to \FNat$ (or $s \colon \FBaire_{\Ip} \to \FNat$) is
equal to its monotonisation with respect to the order $\leq_{\FBaire}$.
\end{lemma}
\begin{proof}
  Immediate from Lemma \ref{lem:ElimZeta} and Corollary \ref{col:MaxCoDomFormalMap}.
\end{proof}

\begin{lemma}\label{lem:pFunction}
For every monotone relation $s \subseteq \FSeq \times \Nat$,
the following are equivalent:
  \begin{enumerate}
    \item\label{lem:pFunction1} $s$ is single valued;
    \item\label{lem:pFunction3} $s^{-} n \downarrow s^{-}m \cov_{\FBaire}
      s^{-} ( n  \downarrow  m )$ for
      all $n,m \in \Nat$;
    \item\label{lem:pFunction4} $s^{-} n \downarrow_{\Ip} s^{-}m
      \cov_{\FBaire_{\Ip}} s^{-} ( n  \downarrow  m )$
      for all $n,m \in \Nat$.
  \end{enumerate}
\end{lemma}
\begin{proof}
  (\ref{lem:pFunction1} $\Rightarrow$ \ref{lem:pFunction3})
  Assume \eqref{lem:pFunction1}. Since $s$ is monotone, we
  have $s^{-} n \downarrow s^{-}m = s^{-} n \mathrel{\cap} s^{-}m$
  for all $n, m \in \Nat$ by \eqref{eq:ConvOrd}.  Then \eqref{lem:pFunction3} is clear.

  \smallskip
  \noindent(\ref{lem:pFunction3} $\Rightarrow$ \ref{lem:pFunction4})
  For each $a \in \FSeq$, define an ideal point $\alpha_{a} \colon \Pt{\FBaire}$ by
  \begin{equation}\label{eq:Pointa}
    \alpha_{a} \defeql \left\{ b \in \FSeq \mid a \leq_{\FBaire} b \right\}
    \cup \big\{ a * \underbrace{0 * \dots * 0}_{n+1} \mid n \in \Nat \big\}.
  \end{equation}
  Now assume \eqref{lem:pFunction3}. 
  We show that $a \cov_{\FBaire_{\Ip}} \left\{ b \right\} \imp a \leq_{\FBaire} b$. 
  Suppose that $a \cov_{\FBaire_{\Ip}} \left\{ b \right\}$. Since
  $a \elm \alpha_{a * 0} \cap \alpha_{a * 1}$, we have $b \elm
  \alpha_{a * 0} \cap \alpha_{a * 1}$, so we must have $a
  \leq_{\FBaire} b$.
  Then \eqref{lem:pFunction4} follows from the fact that
  $\cov_{\FBaire} \subseteq \cov_{\FBaire_{\Ip}}$.

  \smallskip
  \noindent(\ref{lem:pFunction4} $\Rightarrow$ \ref{lem:pFunction1})
  Assume \eqref{lem:pFunction4}. Suppose that $a \mathrel{s} n$
  and $a \mathrel{s} m$. Then 
  \[
    a\cov_{\FBaire_{\Ip}} s^{-} n \downarrow_{\Ip} s^{-}m
    \cov_{\FBaire_{\Ip}} s^{-}( n \downarrow m ).
  \]
  Since $a \elm \alpha_{a}$, we have 
  $\alpha_{a} \meets s^{-}( n \downarrow  m )$. Hence $n = m$.
\end{proof}

Recall from \cite[Chapter 4, Section 8]{ConstMathI} that a subset $U
\subseteq \FSeq$ is a \emph{bar} if every choice sequence has a
neighbourhood which is in $U$. In the setting of positive topology,
we can express the condition of bar as
\[
  \left( \forall \alpha \colon \Pt{\FBaire} \right) \alpha \meets
  U.
\]
The \emph{domain} $\dom(s)$ of a relation $s \subseteq X \times Y$ is
the  subset $s^{-}Y$ of $X$.
\begin{proposition}\label{prop:pFuncMBar}
  For every monotone $s \subseteq \FSeq \times \Nat$,
  the following are equivalent:
  \begin{enumerate}
    \item\label{prop:pFuncMBar1} $s$ is a partial function
      whose domain is a bar;

    \item\label{prop:pFuncMBar2} $s$ is a formal map from
      $\FBaire_{\Ip}$ to $\FNat$.
  \end{enumerate}
\end{proposition}
\begin{proof}
  Fix a monotone relation $s \subseteq \FSeq \times \Nat$.
  Then by Corollary \ref{col:MaxCoDomFormalMap}, $s$ is a
  formal map from $\FBaire_{\Ip}$ to $\FNat$ if and only if
  \ref{def:FormalMap1}, 
  \ref{def:FormalMap2}, and \ref{def:FormalMap3} hold. Condition \ref{def:FormalMap3} is trivial, while condition
  \ref{def:FormalMap2} is equivalent to $s$ being single valued by Lemma
  \ref{lem:pFunction}. Moreover, it is easy to see that condition
  \ref{def:FormalMap1} is
  equivalent to $\dom(s)$ being a bar.
\end{proof}

By Lemma \ref{prop:EquiContOp}, Lemma \ref{lem:Monotonisation}, and
Proposition \ref{prop:pFuncMBar}, we can rephrase the principle
$\ContBN$ as follows:
\begin{description}
  \item[($\ContBN$)]
    If $s \subseteq \FSeq \times \Nat$ is a monotone partial
    function whose domain is a bar, then $s$ is a formal
    map from $\FBaire$ to $\FNat$.
\end{description}
We say that a subset  $U \subseteq \FSeq$ is
  \begin{itemize}
    \item \emph{monotone} if 
      $a \leq_{\FBaire} b \elm U \imp a \elm U$ for all $a, b \in
      \FSeq$;
    \item \emph{inductive} if 
    $
    \left[\left( \forall x \in \Nat \right) \cons{a}{x} \elm U  \right]
    \imp a \elm U
    $
    for all $a \in \FSeq$.
  \end{itemize}
  The \emph{monotone bar induction} \cite[Chapter 4, Section 8]{ConstMathI} is
  the statement:
  \begin{description}
    \item[($\mBI$)]
  For any monotone bar $U \subseteq \FSeq$ and an inductive
  subset $V \subseteq \FSeq$ such that $U \subseteq V$, 
  it holds that  $\nil \elm V$.
  \end{description}

\begin{theorem} \label{thm:ContBEquivBIm}
  The following are equivalent.
  \begin{enumerate}
    \item\label{thm:ContBEquivBIm1} $\mBI$.
    \item\label{thm:ContBEquivBIm2} $\ContBN$.
    \item\label{thm:ContBEquivBIm2prime} $\Cont_{\FBaire, \mathcal{T}}$
      holds for all positive topology $\mathcal{T}$ with the greatest
      positivity.
    \item\label{thm:ContBEquivBIm3} Spatiality of $\FBaire$:
      $a \cov_{\FBaire_{\Ip}} U \imp a \cov_{\FBaire} U$ for all
      $a \in \FSeq$ and $U \subseteq \FSeq$.
    \item\label{thm:ContBEquivBIm4} $\nil \cov_{\FBaire_{\Ip}} U \imp
      \nil \cov_{\FBaire} U$ for all $U \subseteq \FSeq$.
  \end{enumerate}
\end{theorem}
\begin{proof}
  Obviously \eqref{thm:ContBEquivBIm2prime} implies
  \eqref{thm:ContBEquivBIm2}. Moreover, 
  \eqref{thm:ContBEquivBIm3} implies \eqref{thm:ContBEquivBIm2prime}
  by Proposition \ref{prop:SpatImpCont}.
  Thus, it suffices to show 
  (\ref{thm:ContBEquivBIm1} $\Rightarrow$ \ref{thm:ContBEquivBIm4}),
  (\ref{thm:ContBEquivBIm4} $\Rightarrow$ \ref{thm:ContBEquivBIm3}),
  and (\ref{thm:ContBEquivBIm2} $\Rightarrow$
  \ref{thm:ContBEquivBIm1}).

  \smallskip
  \noindent(\ref{thm:ContBEquivBIm1} $\Rightarrow$ \ref{thm:ContBEquivBIm4})
  Assume $\mBI$. Fix $U \subseteq \FSeq$, and suppose that $ \nil
  \cov_{\FBaire_{\Ip}} U$. Then $ \nil
  \cov_{\FBaire_{\Ip}} \downset_{\FBaire} U$, i.e.\ $\downset_{\FBaire} U$
  is a monotone bar. Define a subset $V \subseteq \FSeq$ by
  \[
    V \defeql \left\{ a \in \FSeq \mid a \bcov_{\FBaire}
    \downset_{\FBaire} U\right\}
  \]
  which is an inductive subset containing $\downset_{\FBaire} U$. Thus
  by $\mBI$, we have $\nil \bcov_{\FBaire} \downset_{\FBaire} U$, and
  hence $\nil \cov_{\FBaire} U$ by Lemma \ref{lem:ElimZeta}.

  \smallskip
  \noindent(\ref{thm:ContBEquivBIm4} $\Rightarrow$ \ref{thm:ContBEquivBIm3})
  Assume \eqref{thm:ContBEquivBIm4}. 
  First, one can show that 
  \begin{equation}
    \label{eq:LevelN}
    a \cov_{\FBaire} \left\{ a * b \mid \lh{b} = n \right\}
  \end{equation}
  for all $n \in \Nat$ by induction on $n$. Here,
  $\lh{b}$ denotes the length of a sequence $b$.

  Next, for each $a \in \FSeq$, define $C_{a} \subseteq \FSeq$ by
  \begin{equation}\label{eq:complement}
    C_{a} \defeql \left\{ b \in \FSeq \mid \lh{b} = \lh{a} \amp a \neq b
  \right\}.
  \end{equation}
  Since the equality on $\FSeq$ is decidable, i.e.\ $a = b$ or $\neg(a =
  b)$ for all $a,b \in \FSeq$, we see that
  \begin{equation}
    \label{eq:ClopenCover}
    C_{a} \cup \left\{ a \right\} = \left\{ b \in \FSeq \mid
      \lh{b} = \lh{a}\right\}.
  \end{equation}
  Then, for each $U \subseteq \FSeq$, we have
  \begin{equation}\label{eq:clopen}
    a \cov_{\FBaire} U \leftrightarrow \nil \cov_{\FBaire} U \cup
    C_{a}.
  \end{equation}
  Indeed, if $a \cov_{\FBaire} U$, then by \eqref{eq:LevelN} and
  \eqref{eq:ClopenCover}, we have
  \[
    \nil \cov_{\FBaire}  \left\{ a \right\} \cup C_{a} \cov_{\FBaire} U \cup C_{a}.
  \]
  Conversely, if $\nil \cov_{\FBaire} U \cup C_{a}$ then
  \[
    a \cov_{\FBaire}
    \left( U \cup C_{a} \right) \downset \left\{ a \right\}
    \cov_{\FBaire}
    \left( U \downset \left\{ a \right\} \right) 
    \cup \left( C_{a}\downset \left\{ a \right\} \right)
    \cov_{\FBaire} U.
  \]
  The equivalence \eqref{eq:clopen} holds for the cover
  $\cov_{\FBaire_{\Ip}}$ as well, and admits an analogous proof. For
  example, suppose that $a \cov_{\FBaire_{\Ip}} U$.  Then, for each
  $\alpha \colon \Pt{\mathcal{\FBaire}}$, we have $\alpha \meets C_{a}
  \cup \left\{ a \right\}$ by \eqref{eq:LevelN} and
  \eqref{eq:ClopenCover}.  If $a \elm \alpha$, then $\alpha \meets U$
  because $\alpha$ splits $\cov_{\FBaire}$. Thus $\alpha \meets
  \left( U \cup C_{a} \right)$, and hence $\nil
  \cov_{\FBaire_{\Ip}} U \cup C_{a}$. The converse is also straightforward.

  Now, suppose that
  $a \cov_{\FBaire_{\Ip}} U$. Then $\nil
  \cov_{\FBaire_{\Ip}} U \cup C_{a}$, and thus $\nil \cov_{\FBaire} U \cup C_{a}$
  by the assumption. Hence $a \cov_{\FBaire} U$.

  \smallskip
  \noindent(\ref{thm:ContBEquivBIm2} $\Rightarrow$ \ref{thm:ContBEquivBIm1})
  Assume $\ContBN$. Let $U \subseteq \FSeq$ be a monotone bar
  and $V \subseteq \FSeq$ be an inductive subset such that $U
  \subseteq V$. Define $s_{U} \subseteq  \FSeq \times \Nat$ by
  \[
    a \mathrel{s_{U}} n  \defeqiv  a \elm U \amp n = 0,
  \]
  which is a monotone partial function whose domain is a
  bar. Thus, $s_{U}$ is a formal map from $\FBaire$ to
  $\FNat$ by $\ContBN$. In particular, we have
  \[
    \nil \cov_{\FBaire} \dom(s_{U}) = U = \downset_{\FBaire} U,
  \]
  and so $\nil \bcov_{\FBaire}  U$ by Lemma \ref{lem:ElimZeta}.
  Therefore $\nil \elm V$ by induction on $\bcov_{\FBaire}$.
\end{proof}
\begin{remark}
  The equivalence between spatiality of $\FBaire$ and $\mBI$ is well known
  (cf.\ Schuster and Gambino \cite[Proposition
4.1]{SchusterGambinoSpatiality}; see Fourman and Grayson \cite[Theorem
3.4]{FormalSpace} for an impredicative result).
\end{remark}

\subsection{Variety of continuity principles}\label{sec:VariContBaire}
We introduce some variants of $\ContBN$ by imposing some restrictions
on the relation $s \subseteq \FSeq \times \Nat$. We study
the connections between these
principles and some well-known variants of bar induction.

\begin{definition}
  $\ContBND$ is a variant of $\ContBN$ 
  formulated with respect to a monotone partial function with a
  \emph{decidable} domain:
\begin{description}
  \item[($\ContBND$)]
    If $s \subseteq \FSeq \times \Nat$ is a monotone partial
    function whose domain is a decidable bar, then $s$ is a formal
    map from $\FBaire$ to $\FNat$.
\end{description}
\end{definition}
Apart from the distinction between operations and functions (total
and single-valued relations) in the Minimalist Foundation, 
it is easy to see that monotone partial functions from $\FSeq$ to
$\Nat$ whose domains
are decidable bars bijectively correspond
to \emph{neighbourhood functions}; see Beeson~\cite[Chapter 6, Section
7.3]{BeesonFoundationConstMath} or Troelstra and van
Dalen~\cite[Chapter 4, Section 6.8]{ConstMathI}.

The \emph{decidable bar induction} is the statement:%
\footnote{Brouwer seems to have introduced bar induction in this form
  \cite[Section 2]{BrouwerDomainsofFunctions}; see also Kleene and
  Vesley \cite[Chapter I, Section 6]{KleeneVesley}.}
\begin{description}
  \item[($\dBI$)]
For any decidable bar $U \subseteq \FSeq$ and an inductive
subset $V \subseteq \FSeq$ such that $U \subseteq V$, 
it holds that  $\nil \elm V$.

\end{description}
In $\dBI$, we may assume that a bar is monotone since the downward
closure $\downset_{\FBaire} U$ of a decidable subset $U \subseteq
\FSeq$ is decidable (cf.\ Troelstra and van Dalen \cite[Exercise
4.8.10]{ConstMathI}).

\begin{proposition} \label{prop:ContBDEquivBId}
  The following are equivalent.
  \begin{enumerate}
    \item\label{prop:ContBDEquivBId1} $\dBI$.
    \item\label{prop:ContBDEquivBId2} $\ContBND$.
    \item\label{prop:ContBDEquivBId3} $a \cov_{\FBaire_{\Ip}} U
      \imp a \cov_{\FBaire} U$ for all $a \in \FSeq$ and decidable $U \subseteq
      \FSeq$.
    \item\label{prop:ContBDEquivBId4} $\nil \cov_{\FBaire_{\Ip}} U
      \imp \nil \cov_{\FBaire} U$ for all decidable $U \subseteq
      \FSeq$.
  \end{enumerate}
\end{proposition}
\begin{proof}
  We show $(\ref{prop:ContBDEquivBId4} \Rightarrow
  \ref{prop:ContBDEquivBId2})$. Assume \eqref{prop:ContBDEquivBId4}, and let
  $s \subseteq \FSeq \times \Nat$ be a monotone partial function whose
  domain is a decidable bar. Then
  $\nil \cov_{\FBaire} \dom(s)$ by assumption; hence
  $s$ is a formal map  $s \colon \FBaire \to \FNat$ by 
  Corollary~\ref{col:MaxCoDomFormalMap} and Lemma~\ref{lem:pFunction}.

  The proof of the other equivalences are analogous to those of Theorem
  \ref{thm:ContBEquivBIm}.
  Note that for each $a \in \FSeq$ and decidable $U \subseteq \FSeq$,
  the union $U \cup C_{a}$ is decidable,
  where $C_{a}$ is defined as in \eqref{eq:complement}.
\end{proof}
\begin{remark}
  The equivalence between $\dBI$ and $\ContBND$ is analogous to
  Proposition 8.14 (i) in Troelstra and van Dalen \cite[page
  230]{ConstMathI}, where it is shown that $\dBI$ is equivalent
  to identification of the class of neighbourhood functions
  and that of inductively defined neighbourhood functions.
\end{remark}

The second variant of $\ContBN$ is based on the assumption that
the domain of 
a partial function $s \subseteq \FSeq \times \Nat$
is recursively enumerable. The logical
counterpart of this assumption is that the domain of $s$ is a
$\Sigma^{0}_{1}$ subset of $\FSeq$.  Here, a subset $U \subseteq
\FSeq$ is $\Sigma^{0}_{1}$ if there exists a decidable subset $D
\subseteq \FSeq \times \Nat$ such that 
\[
  a \elm U \leftrightarrow \left( \exists n \in \Nat \right) (a,n) \elm D.
\]
\begin{definition}
 $\ContSigmaOne$ is a variant of $\ContBN$ 
 formulated with respect to a monotone partial function whose domain
 is a $\Sigma^{0}_{1}$ subset of $\FSeq$:
\begin{description}
  \item[($\ContSigmaOne$)]
    If $s \subseteq \FSeq \times \Nat$ is a monotone partial
    function whose domain is a $\Sigma^{0}_{1}$ bar, then $s$ is a formal
    map from $\FBaire$ to $\FNat$.
\end{description}
\end{definition}
The principle $\ContSigmaOne$ is related to
the corresponding principle of bar induction.
\emph{$\Sigma^{0}_{1}$ monotone bar induction} is the statement:
\begin{description}
  \item[($\BISigmaOne$)]
    For any monotone $\Sigma^{0}_{1}$ bar $U \subseteq \FSeq$ and an inductive
  subset $V \subseteq \FSeq$ such that $U \subseteq V$, 
  it holds that  $\nil \elm V$.

\end{description}
In $\BISigmaOne$, we cannot drop the condition of monotonicity on
bars; otherwise the non-constructive principle $\LPO$ (the limited
principle of omniscience) would be derivable,
which is not acceptable in constructive mathematics (cf.\ Troelstra and van Dalen
\cite[Exercise 4.8.11]{ConstMathI}).
\begin{proposition} \label{prop:ContSigmaOnequivBISimgaOne}
  The following are equivalent.
  \begin{enumerate}
    \item $\BISigmaOne$.
    \item $\ContSigmaOne$.
    \item $a \cov_{\FBaire_{\Ip}} U
      \imp a \cov_{\FBaire} U$ for all $a \in \FSeq$ and
      $\Sigma^{0}_{1}$-subset $U \subseteq
      \FSeq$.
    \item $\nil \cov_{\FBaire_{\Ip}} U
      \imp \nil \cov_{\FBaire} U$ for all $\Sigma^{0}_{1}$-subset $U \subseteq
      \FSeq$.
  \end{enumerate}
\end{proposition}
\begin{proof}
  The proof is analogous to that of Proposition  \ref{prop:ContBDEquivBId}.
  Note that for each $a \in \FSeq$ and $\Sigma^{0}_{1}$ subset $U
  \subseteq \FSeq$, the union $U \cup C_{a}$ is also
  $\Sigma^{0}_{1}$.
\end{proof}
In fact, these $\Sigma^{0}_{1}$ variants are equivalent to the
decidable counterparts. Ishihara \cite[Proposition
16.15]{ConstRevMatheCompactness} already observed an analogous fact for the
fan theorem.
In the following, we write $\langle n,m
\rangle$ for a fixed coding of pairs of numbers $n,m \in \Nat$
and write $j_{0},
j_{1}$ for the projections of pairs. Without loss of generality, we
 assume that the coding is surjective and satisfies
$n,m \leq \langle n,m \rangle$. 
The proof of the following proposition is based on  Ishihara \cite[Proposition
16.15]{ConstRevMatheCompactness}.
\begin{proposition}\label{prop:EquivSigmaD}
$\BISigmaOne$ and $\dBI$ are equivalent.
\end{proposition}
\begin{proof}
It suffices to show that $\dBI$ implies $\BISigmaOne$. To this end, we
show that every monotone $\Sigma^{0}_{1}$ bar contains a decidable
bar.  Let $U \subseteq \FSeq$ be a monotone $\Sigma^{0}_{1}$ bar, and
let $D \subseteq \FSeq \times \Nat$ be a decidable subset such that $a
\elm U \leftrightarrow \left( \exists n \in \Nat \right)(a,n) \elm
D$ for all $a \in \FSeq$. Define $V \subseteq \FSeq$ by
\[
  V \defeql \left\{ a \in \FSeq \mid
    (\overline{a}j_{0}(\lh{a}), j_{1}(\lh{a})) \elm
    D\right\}
  \]
where $\overline{a}n \;(n < \lh{a})$ is the initial segment
of $a$ of length $n$. Obviously, $V$ is decidable. Moreover, 
since $a \leq_{\FBaire} \overline{a}j_{0}(\lh{a})$ and 
$U$ is monotone, $V$ is contained in $U$. 
Lastly, to see that $V$ is a bar, let $\alpha \colon
\Pt{\mathcal{\FBaire}}$. Since $U$ is a bar, there exists $a \elm
\alpha$ such that $a \elm U$. Since $U$ is $\Sigma^{0}_{1}$, there exists
$n \in \Nat$ such that $D(a,n)$. By \eqref{eq:LevelN}, there exists $b
\elm \alpha$  such that $\lh{b} = \langle \lh{a},n \rangle$.
Then $b \elm V$, and hence $V$ is bar.
\end{proof}

Another variant of $\ContBN$ is the following principle.
\begin{definition}\label{def:ContPiOne}
  $\ContPiOne$ is a variant of $\ContBN$
  formulated with respect to a monotone partial function whose domain
  is a $\Pi^{0}_{1}$ subset of $\FSeq$:
\begin{description}
  \item[($\ContPiOne$)]
    If $s \subseteq \FSeq \times \Nat$ is a monotone partial
    function whose domain is a $\Pi^{0}_{1}$ bar, then $s$ is a formal
    map from $\FBaire$ to $\FNat$.
\end{description}
Here a subset $U \subseteq \FSeq$ is $\Pi^{0}_{1}$ if 
there exists a decidable subset $D \subseteq \FSeq \times \Nat$
such that 
\[
  a \elm U \leftrightarrow \left( \forall n \in \Nat \right) (a,n) \elm D.
\]
\end{definition}
A  monotone partial function with a
$\Pi^{0}_{1}$ domain arises as a \emph{modulus} of a continuous
map $(g,s) \colon \Pt{\FBaire} \to \Pt{\FNat}$. 
This can be made precise as follows: let $\Image{s} \colon \Pt{\FBaire} \to
\Pt{\FNat}$ be the continuous map induced by a relation
$s \subseteq \FSeq \times \Nat$. Define a monotone partial function
$s'
\subseteq \FSeq \times \Nat$ by
\[
  a \mathrel{s'} n \defeqiv \left( \forall b \in \FSeq \right)
    \Image{s}(\alpha_{a})  = \Image{s}(\alpha_{a*b}) \amp n \elm
    \Image{s}(\alpha_{a}),
\]
where $\alpha_{a}$ is defined by \eqref{eq:Pointa}.
Note that
\[
  \dom(s') = \left\{ a \in \FSeq \mid \left( \forall b \in \FSeq \right)
    \Image{s}(\alpha_{a})  = \Image{s}(\alpha_{a*b})\right\},
  \]
which is a $\Pi^{0}_{1}$ subset of $\FSeq$ under
a suitable coding of $\FSeq$ in $\Nat$. We claim that
\begin{equation}\label{eq:MonotonePFuncPiDom}
  \left( \forall \alpha \colon \Pt{\FBaire} \right) \Image{s'}(\alpha) =
  \Image{s}(\alpha).
\end{equation}
To see this, fix $\alpha \colon \Pt{\FBaire}$. Let $n \elm
\Image{s'}(\alpha)$, and let $a \elm \alpha$ be such that $a
\mathrel{s'} n$.  Let $m \elm \Image{s}(\alpha)$ be the unique
element of $\Image{s}(\alpha)$, and $b \elm \alpha$ be such that $b
\mathrel{s} m$. Since $\alpha$ is filtering, we have either $a
\leq_{\FBaire}
b \amp b \neq a$ or $b \leq_{\FBaire} a$. In the former case, we have $m
\elm \Image{s}(\alpha_{a})$, and hence $n = m$ because
$\Image{s}(\alpha_{a})$ is a singleton. In the latter case, we have
$\Image{s}(\alpha_{a}) =  \Image{s}(\alpha_{b})$ because
$a \elm \dom(s')$. Since $n \elm \Image{s}(\alpha_{a})$ and $m
\elm \Image{s}(\alpha_{b})$, we have $n = m$. Therefore
$\Image{s'}(\alpha) \subseteq \Image{s}(\alpha)$. The converse
inclusion of \eqref{eq:MonotonePFuncPiDom} is obvious.
In summary, we have the following proposition.

\begin{proposition}\label{prop:ContMapMonotonePFuncPiDom}
For every continuous map $(\Image{s},s) \colon \Pt{\FBaire} \to
\Pt{\FNat}$, there exists a  monotone partial function $s' \subseteq
\FSeq \times \Nat$ with a $\Pi^{0}_{1}$ domain such that
$\Image{s} = \Image{s'}$.
\end{proposition}
Note that Proposition \ref{prop:ContMapMonotonePFuncPiDom} does not imply that
$\ContPiOne$ is equivalent to $\ContBN$ since equality of
formal maps from $\FBaire_{\Ip}$ to $\FNat$ is weaker than equality
of formal maps from $\FBaire$ to $\FNat$.

The principle of bar induction that corresponds to $\ContPiOne$ is
\emph{$\Pi^{0}_{1}$ monotone bar induction}:
\begin{description}
  \item[($\BIPiOne$)]
    For any monotone $\Pi^{0}_{1}$ bar $U \subseteq \FSeq$ and an inductive
  subset $V \subseteq \FSeq$ such that $U \subseteq V$, 
  it holds that  $\nil \elm V$.

\end{description}
In $\BIPiOne$,
we cannot drop the condition of monotonicity on bars;
otherwise the non-constructive principle $\LLPO$ (the lesser
limited principle of omniscience) would be
derivable (cf.\ Kawai \cite[Proposition 7.1]{KawaiUnifContBaire}).
\begin{proposition} \label{prop:ContPiOnequivBIPiOne}
  The following are equivalent.
  \begin{enumerate}
    \item $\BIPiOne$.
    \item $\ContPiOne$.
    \item $a \cov_{\FBaire_{\Ip}} U
      \imp a \cov_{\FBaire} U$ for all $a \in \FSeq$ and 
      monotone
      $\Pi^{0}_{1}$-subset $U \subseteq
      \FSeq$.
    \item $\nil \cov_{\FBaire_{\Ip}} U
      \imp \nil \cov_{\FBaire} U$ for all monotone $\Pi^{0}_{1}$-subset $U \subseteq
      \FSeq$.
  \end{enumerate}
\end{proposition}
\begin{proof}
  The proof is analogous to that of Proposition
  \ref{prop:ContSigmaOnequivBISimgaOne}. 
\end{proof}
By Proposition \ref{prop:ContMapMonotonePFuncPiDom} and Proposition
\ref{prop:ContPiOnequivBIPiOne},  the principle
$\BIPiOne$ implies that every continuous map $(\Image{s},s) \colon
\Ip( \FBaire ) \to \Ip(\FNat)$ is induced by some formal map $s' \colon
\FBaire \to \FNat$. That the converse of this holds is essentially
known (cf.\ Kawai \cite[Theorem
3.9]{KawaiFContOnBaire}). For the sake of completeness, we sketch the
proof using our notion of continuous map.
\begin{theorem} \label{thm:PiBiEquivInducedFMap}
  The following are equivalent.
  \begin{enumerate}
    \item\label{thm:PiBiEquivInducedFMap1} $\BIPiOne$.
    \item\label{thm:PiBiEquivInducedFMap2} For any continuous map 
    $(\Image{s},s) \colon \Ip( \FBaire ) \to \Ip(\FNat)$, there exists
    a formal map $s' \colon \FBaire \to \FNat$ such that
    $\Image{s} = \Image{s'}$.
  \end{enumerate}
\end{theorem}
\begin{proof}
  It suffices to show that \eqref{thm:PiBiEquivInducedFMap2}
  implies \eqref{thm:PiBiEquivInducedFMap1}.
  Assume \eqref{thm:PiBiEquivInducedFMap2}. Let $U \subseteq \FSeq$ be a
  $\Pi^{0}_{1}$ monotone bar, and $V \subseteq \FSeq$ be an
  inductive subset such that $U \subseteq V$. Then, there exists a
  decidable subset $D \subseteq \FSeq \times \Nat$ such that $a \elm
  U \leftrightarrow \left( \forall n \in \Nat \right)D(a, n)$ for all
  $a \in \FSeq$. Since $U$ is monotone, the right hand side is
  equivalent to 
  $\left( \forall b \in \FSeq \right)\left( \forall n \in \Nat
  \right)D(a * b, n)$ for all $a \in \FSeq$. Define a decidable subset
  $\overline{D} \subseteq \FSeq$ by
  \[
    \overline{D}(a) \defeqiv \lh{a} > 0 \imp D(\head{a},\tail{a})
  \]
  where $\head{a}$ is obtained from $a$ by omitting the last
  entry $\tail{a}$ of $a$. Furthermore, define $\overline{U}
  \subseteq \FSeq$ by
  \[
    \overline{U}(a) \defeqiv U(a) \amp \overline{D}(a).
  \]
  Then, it is straightforward to see that
  $\overline{U}$ is a monotone bar and that 
  \[
    \overline{U}(a) \leftrightarrow \left( \forall b \in \FSeq \right)
    \overline{D}(a*b).
  \]
  Note that $\overline{U}$ is a $\Pi^{0}_{1}$  subset of $\FSeq$.
  Define a relation $s \subseteq \FSeq \times \Nat$ by
  \begin{multline*}
    a \mathrel{s} n \defeqiv 
    \overline{U}(a) \amp
    \Bigl[ \left(  n < \lh{a}
      \amp \neg \overline{D}(\overline{a}n) \amp \left( \forall m <
      \lh{a} \right) \left[n < m \imp \overline{D}(\overline{a}m)
      \right]\right) \Bigr . \\
    \Bigl .
    \vee \left( \left( \forall m
      < \lh{a} \right) \overline{D}(\overline{a}m) \amp n = 1
    \right)\Bigr].
  \end{multline*}
  It is easy to see that $s$ is a monotone partial function with
  domain $\overline{U}$. By the assumption, there exists a formal map
  $s' \colon \FBaire \to \FNat $ such that $\Image{s'} =
  \Image{s}$. Define a relation $\overline{s'} \subseteq \FSeq
  \times \Nat$ by 
  \[
    a \mathrel{\overline{s'}} n \defeqiv n < \lh{a} \amp \left( \exists b
    \geq_{\FBaire} a \right) b \mathrel{s'} n.
  \]
  By \eqref{eq:LevelN}, we see that $\overline{s'}$ and $s'$ are equal
  as formal maps from $\FBaire$ to $\FNat$.
  Thus $\Image{\overline{s'}} = \Image{s'}$. Then, it is
  straightforward to show that $\dom(\overline{s'}) \subseteq
  \overline{U}$. Hence $\nil \cov_{\FBaire} U$, and therefore $\nil
  \elm V$ by Lemma \ref{lem:ElimZeta}.
\end{proof}

\subsection{Relativisation to spreads}\label{sec:Spread}
In this section, we relativise our continuity
principle to spreads. This is in accord
with Brouwer \cite{BrouwerDomainsofFunctions},
who considered continuity of functions on 
choice sequences which belong to a fixed spread.

\begin{definition}
  A \emph{spread} is a decidable tree of natural numbers in which every node
  has an extension. Specifically, a spread is an inhabited decidable subset
  $U \subseteq \FSeq$ such that
  \begin{enumerate}
    \item $a \elm U  \amp a \leq_{\FBaire} b \imp b \elm U$,
    \item $a \elm U \imp \left( \exists n \in \Nat \right)
      \cons{a}{n} \elm U$
  \end{enumerate}
  for all $a,b \in S$ and $U \subseteq S$.
  Note that every spread enters $\pos_{\FBaire}$.
  We say that a choice sequence $\alpha \colon \Pt{\FBaire}$
  \emph{belongs to} a
spread $U$ if $\alpha \subseteq U$, i.e.\
  every choice made by $\alpha$ is in the spread.
\end{definition}
  Every spread $U$ determines a cover $\cov_{U}$ and a
  positivity $\pos_{U}$ on $\FSeq$ by
  \begin{align}\label{eq:SubSpaSpread}
      a \cov_{U} V &\defeqiv a \cov_{\FBaire} \neg U \cup V  &
      a \pos_{U} V &\defeqiv a \pos_{\FBaire} U \cap V 
  \end{align}
  where $\neg U \defeql \left\{ a \in \FSeq \mid \neg \left( a \elm U
  \right) \right\}$.
\begin{lemma}\label{lem:SubSpaSpread}
  Let $U$ be a spread.
  \begin{enumerate}
    \item\label{lem:SubSpaSpread1} $\FBaire_{U} \defeql (\FSeq,
      \cov_{U}, \pos_{U})$ is a positive topology.
    \item\label{lem:SubSpaSpread3} $\pos_{U}$ is the 
      greatest
      positivity that is compatible with $\cov_{U}$.
    \item\label{lem:SubSpaSpread4}
      $\alpha \colon \Pt{\FBaire_{U}}$
      if and only if $\alpha$ is a choice sequence belonging to $U$.
  \end{enumerate}
\end{lemma}
\begin{proof}
  \eqref{lem:SubSpaSpread1} We leave the reader to check that
  $\cov_{U}$ is a cover and $\pos_{U}$ is a positivity. We show that
  $\cov_{U}$ and $\pos_{U}$ are compatible. Let $a \in \FSeq$ and
  $V,W \subseteq \FSeq$.  Suppose that $a \cov_{U}V$ and $a \pos_{U}
  W$. Since $\cov_{\FBaire}$ and $\pos_{\FBaire}$ are compatible, we
  have $(\neg U \cup V) \pos_{\FBaire} U \cap W$.  Either $\neg U
  \pos_{\FBaire} U \cap W$ or $V \pos_{\FBaire} U \cap W$. The former
  case leads to a contradiction.  Hence $V \pos_{\FBaire} U \cap W$,
  i.e.\ $V \pos_{U} W$, as required.

\smallskip
\noindent\eqref{lem:SubSpaSpread3} Let $\pos'$ be a positivity on
$\FSeq$ which is compatible with $\cov_{U}$.
We show that
\[
  \frac{a \pos' V \quad
    \left( \forall b \in \FSeq \right)\left( b \pos' V \imp
  b \elm W\right)} {a \pos_{\FBaire} W}
\]
for all $V, W \subseteq \FSeq$. 
To this end, fix $V, W \subseteq \FSeq$, and put $V' = \left\{ a \in
\FSeq \mid a \pos' V \right\}$. Suppose that $a \pos' V$ and 
$\left( \forall b \in \FSeq \right) b \pos' V \imp b \elm W$.
This is equivalent to $a \elm V'$ and $V' \subseteq W$.
Since $\cov_{\FBaire} \subseteq \cov_{U}$, $\pos'$ is compatible with
$\cov_{\FBaire}$. Thus $\pos' \subseteq \pos_{\FBaire}$. Since 
$a \elm V'$ implies $a \pos' V'$, we have  $a
\pos_{\FBaire} V'$. Hence $a \pos_{\FBaire} W$.

Now, to see that $\pos' \subseteq \pos_{U}$,  it suffices to show that $a \pos' V \imp a \elm U
\cap V$ for all $a \in \FSeq$ and $V \subseteq \FSeq$. Suppose that $a
\pos' V$. Then, trivially $a \elm V$.  Moreover, either $a \elm U$ or
$a \elm \neg U$.  If $a \elm \neg U$, then $a \cov_{U} \emptyset$, and
by compatibility we obtain $\emptyset \pos' V$, a contradiction.
Hence $a \elm U$, and therefore $a \elm U \cap V$.

\smallskip
\noindent\eqref{lem:SubSpaSpread4} First, suppose that $\alpha \colon
\Pt{\FBaire_{U}}$. It is straightforward to show that $\alpha$
is an ideal point of $\FBaire$. Moreover, since $\neg U \cov_{U} \emptyset$, we have 
$\alpha \cap \neg U = \emptyset$. As $U$ is decidable, we obtain
$\alpha \subseteq \neg \neg U = U$. Conversely, suppose that $\alpha$
belongs to $U$. We show that $\alpha$ enters
$\pos_{U}$. Let $a \in \FSeq$ and $V \subseteq \FSeq$, and suppose
that $a \elm \alpha \subseteq V$. Then, $a \elm \alpha \subseteq V
\cap U$. Since $\alpha$ enters $\pos_{\FBaire}$, we have $a
\pos_{\FBaire} V \cap U$, i.e.\ $a \pos_{U} V$, as required.
The other properties of ideal points are easily checked.
\end{proof}
We often identify a spread $U \subseteq \FSeq$ with the topology
$\FBaire_{U}$.

\begin{lemma}
  \label{lem:ImgSpread}
  Every spread is a retract of $\FBaire$.
\end{lemma}
\begin{proof}
  This is well known in point-set topology (cf.\ Troelstra and van
  Dalen \cite[Chapter 4, Lemma 1.4]{ConstMathI}). Here, we give a
  pointfree proof.
  
  Fix a spread $U \subseteq \FSeq$.  Clearly, the identity relation
  $\id_{\FSeq}$ on $\FSeq$ is a formal map from $\FBaire_{U}$ to
  $\FBaire$. Its retraction $s \colon \FBaire \to
  \FBaire_{U}$ is inductively defined as follows:
  \begin{gather*}
    \frac{}{\nil \mathrel{s} \nil} \qquad
    \frac{a \mathrel{s} b \quad b * n \elm U}
    {a * n \mathrel{s} b * n} \qquad
    \frac{a \mathrel{s} b \quad \neg (b * n \elm U)}
    {a * n \mathrel{s} b * l}
\end{gather*}
In the last rule, $l$ is the least number such that
$b * l \elm U$. The following properties of $s$ directly follow from
the definition:
\begin{enumerate}
  \item $s$ is a function.
  \item $a \elm U \leftrightarrow a \mathrel{s} a$.
  \item $a \mathrel{s} b \imp \lh{a} = \lh{b}$.
  \item $a * n \mathrel{s} b * m \imp a \mathrel{s} b$.
  \item $a \leq_{\FBaire} b \amp a \mathrel{s} c \amp b \mathrel{s} d \imp c \leq_{\FBaire}  d$.
\end{enumerate}
Then, it is straightforward to show that $s$ is a
formal map from $\FBaire$ to $\FBaire$. We show that $\im[s] =
\FBaire_{U}$, i.e.\ $\cov_{s} = \cov_{U}$ and $\pos_{s} = \pos_{U}$
(cf.\ \eqref{eq:Image} and \eqref{eq:SubSpaSpread}).

\smallskip
\noindent $\cov_{s} = \cov_{U}$: Fix $a \in \FSeq$ and $V \subseteq
\FSeq$. First, suppose that $a \cov_{s} V$, i.e.\ $s^{-}a
\cov_{\FBaire} s^{-}V$. Either $a \elm U$ or $a \elm \neg U$. In the latter case, we have $a \cov_{U} V$. In
  the former case, we have $a \mathrel{s} a$. Thus, $a \cov_{\FBaire}
  s^{-}V$. Let $b \elm s^{-}V$. Either $b \elm U$ or $b \elm \neg U$. 
  If $b \elm U$, then $b \elm V$. 
  Hence $s^{-} V \subseteq \neg U
  \cup V$, and so $a \cov_{U} V$. 

  The converse is proved by induction on $\cov_{\FBaire}$. Suppose that $a
  \cov_{U}V$ is derived by $\eta$-rule. Then either $a \elm \neg U$ or
  $a \elm V$. If $a \elm \neg U$, then $s^{-}a = \emptyset
  \cov_{\FBaire} s^{-}V$, and thus $a \cov_{s} V$. If $a \elm V$, then
  trivially $a \cov_{s}V$. The proof for the other rules follow
  easily from induction hypothesis and the above mentioned properties of $s$.

  \smallskip
  \noindent $\pos_{s} = \pos_{U}$: Since $\cov_{s} = \cov_{U}$ and
  $\pos_{U}$ is the greatest positivity compatible with $\cov_{U}$, it
  suffices to show that $\pos_{U} \subseteq \pos_{s}$. Suppose that $a
  \pos_{U} V$.  We must show that $a \pos_{s} V$, i.e.\ $s^{-}a
  \pos_{\FBaire} s^{*}V$.  But $a \elm U$, and so $a \mathrel{s} a$.
  Since $U \cap V \subseteq s^{*}V$, we have $a \pos_{\FBaire}
  s^{*}V$.

  \smallskip
  Therefore, $s$ is a formal map from $\FBaire$ to $\FBaire_{U}$.
  Then, it is straightforward to show that the composition of
  $\id_{\FSeq} \colon \FBaire_{U} \to \FBaire$ and $ s \colon \FBaire
  \to \FBaire_{U}$ is the identity.
\end{proof}

\begin{proposition}\label{prop:SpreadPoPC}
  For every spread $U$ and a positive topology $\mathcal{T}$,
  the principle $\Cont_{\FBaire, \mathcal{T}}$ implies
  $\Cont_{\FBaire_U,\mathcal{T}}$.
\end{proposition}
\begin{proof}
  Immediate from Lemma \ref{lem:ImgSpread} and Proposition \ref{prop:ImagePoPC}.
\end{proof}
In particular, by Lemma \ref{lem:SubSpaSpread},  Proposition
\ref{prop:SpreadPoPC}, and Theorem \ref{thm:ContBEquivBIm}, $\mBI$
implies that any relation $s \subseteq \FSeq \times \FSeq$ which maps
every choice sequence belonging to a spread $U$ to another spread $V$
is a formal map from $\FBaire_{U}$ to $\FBaire_V$.

The proof of the following is analogous to that of Theorem
\ref{thm:ContBEquivBIm}.
\begin{proposition}
  $\Cont_{\FBaire_{U},\FNat}$ is equivalent to spatiality of $\FBaire_{U}$.
\end{proposition}

\emph{The formal Cantor space} is the binary spread $\BSeq$. It is
well known that the fan theorem and spatiality of the formal Cantor
space are equivalent (cf.\ Schuster and Gambino \cite[Proposition
4.3]{SchusterGambinoSpatiality}; see Fourman and Grayson \cite[Theorem
3.4]{FormalSpace} for an impredicative result). 
Thus, the fan theorem and the principle of continuity for the formal
Cantor space and the formal topology of natural numbers are equivalent.

\section{Continuity on real numbers }\label{sec:ContReal}
We consider the continuity principle for the formal
unit interval $\FUInt$.

\begin{definition}
  \label{eg:FormalReal}
  Let $\mathbb{Q}$ be the set of rational numbers, and let
  \[
    S_{\FReal} \defeql \left\{ (p,q) \in \mathbb{Q} \times \mathbb{Q} \mid p < q
  \right\}.
\]
  Define a preorder $\leq_{\FReal}$ and a strict order $<_{\FReal}$ on
  $S_{\FReal}$ by
  \begin{align*}
    (p,q) &\leq_{\FReal} (p',q') \defeqiv p' \leq p \amp q \leq q', \\
    (p,q)  &<_{\FReal}   (p',q') \defeqiv p' < p \amp q < q'.
  \end{align*}
  The \emph{formal topology of real numbers}
  $\FReal$ is a positive topology
  $\FReal = (S_{\FReal}, \cov_{\FReal}, \pos_{\FReal} )$ where
  $\cov_{\FReal}$ is inductively generated
  by the following rules
  \begin{gather*}
    \frac{(p,q)  \elm U}{(p,q) \cov_{\FReal} U} \qquad
    \frac{(p,q)  \leq_{\FReal} (p',q') \cov_{\FReal}U}{(p,q)
      \cov_{\FReal} U} \qquad
    \frac{\left( \forall (p',q') <_{\FReal} (p,q) \right)  (p',q')
    \cov_{\FReal}U}{(p,q) \cov_{\FReal} U} \\[.5em]
    \frac{p < p' < q' < q \quad (p,q') \cov_{\FReal} U \quad (p',q)
    \cov_{\FReal} U }{(p,q) \cov_{\FReal} U}
\end{gather*}
and $\pos_{\FReal}$ is the greatest positivity compatible with $\cov_{\FReal}$.
The \emph{formal unit interval} $\FUInt = (S_{\FReal},
\cov_{\FUInt}, \pos_{\FUInt})$ is a positive topology where
$\cov_{\FUInt}$ is generated by the rules of $\cov_{\FReal}$
together with the following rules:
  \begin{gather*}
    \frac{p < q < 0}{(p,q) \cov_{\FUInt} \emptyset} \qquad
    \frac{1 < p < q}{(p,q) \cov_{\FUInt} \emptyset}
\end{gather*}
The positivity $\pos_{\FUInt}$ is again the 
greatest one compatible
with $\cov_{\FUInt}$. 
\end{definition}
The collection of ideal points of $\FReal$ is isomorphic to
the Dedekind reals (Negri and Soravia \cite[Proposition
9.2]{NegriSoravia99}; see also Fourman and Grayson \cite[Example 1.2
(4)]{FormalSpace}). 
Similarly, $\Pt{\FUInt}$ is isomorphic to the unit interval of the
Dedekind reals, which is denoted by $\left[ 0,1 \right]$.  It is well
known that  $\FUInt$ is compact: every cover has a finite subcover
(cf.\ Cederquist and Negri \cite[Theorem
23]{HeineBorelFSpa}).\footnote{These results about the formal topology
  of real numbers are
  based on the notion of formal topology
  \cite{Sambin:intuitionistic_formal_space}, but they carry over to
our setting because the cover of $\FReal$ is inductively
generated and the positivity is the greatest
one compatible with the cover.}

Finite covers of $\cov_{\FReal_{\Ip}}$ and $\cov_{\FReal}$
coincide as in the following lemma (cf.\ Palmgren \cite[Lemma 10.6]{PalmgrenRealFTop}).
\begin{lemma} \label{lem:FinCovR}
  For any $(p,q) \in S_{\FReal}$ and a finite subset $U \subseteq
  S_{\FReal}$, we have
  \[
    (p,q) {\cov_{\FReal_{\Ip}}} U \leftrightarrow  (p,q) \cov_{\FReal} U.
  \]
\end{lemma}
\begin{proof}
  By induction on the size of $U$.
\end{proof}

Let $\ContUInt$ be the statement:
$\Cont_{\FUInt,\mathcal{T}}$ holds for all positive topology
$\mathcal{T}$ equipped with the greatest positivity compatible with
its cover.

\begin{theorem}\label{thm:SpatialityFR}
  The following are equivalent.
  \begin{enumerate}
    \item\label{thm:SpatialityFR1} $\ContUInt$.
    \item\label{thm:SpatialityFR2} $\UInt$ is a compact subset of
      $\Pt{\FReal}$, i.e.\ every open cover has a finite subcover
      (Heine--Borel covering theorem).
    \item\label{thm:SpatialityFR4} $\FReal$ is spatial.
    \item\label{thm:SpatialityFR3} $\FUInt$ is spatial.
  \end{enumerate}
\end{theorem}
\begin{proof}
  The equivalence between
  \eqref{thm:SpatialityFR2},
  \eqref{thm:SpatialityFR4}, and
  \eqref{thm:SpatialityFR3} are well known in locale
  theory (cf.\ Fourman and Grayson
  \cite[Theorem 4.10]{FormalSpace}). We give a predicative proof
  below.
  \smallskip

  \noindent(\ref{thm:SpatialityFR1} $\Rightarrow$
  \ref{thm:SpatialityFR2}) 
  Let $U \subseteq S_{\FReal}$ be a cover of  $\UInt$, i.e.\
  $S_{\FReal} \cov_{\FUInt_{\Ip}} U$. Define $s \subseteq
  S_{\FReal} \times \Nat$ by
  \[
    a \mathrel{s} n \defeqiv a \elm \downset_{\leq_{\FReal}}U \amp n =
    0,
  \]
  where $\downset_{\leq_{\FReal}}U$ is 
  the downward closure of $U$ with respect to $\leq_{\FReal}$.
  Note that $s$ is a monotone partial function with domain
  $\downset_{\leq_{\FReal}} U$. By the similar argument as in
  Lemma \ref{lem:pFunction}, we see that $s$ is a formal map $s \colon
  \FUInt_{\Ip} \to \FNat$.
   By $\ContUInt$, $s$ is a formal map $s \colon \FUInt \to \FNat$ as
   well.
   Hence $S_{\FReal} \cov_{\FUInt} U$.  Since $\FUInt$ is compact,
   there exists a finite $V \subseteq U$ such that $S_{\FReal}
   \cov_{\FUInt} V$. Therefore $S_{\FReal} \cov_{\FUInt_{\Ip}} V$.
  \smallskip

  \noindent(\ref{thm:SpatialityFR2} $\Rightarrow$
  \ref{thm:SpatialityFR4}) Assume that
  $\left[ 0,1 \right]$ is compact, and suppose that
  $(p,q) \cov_{\FReal_{\Ip}} U$. Let $(p',q') <_{\FReal} (p,q)$. Then
  $\left[p',q' \right] \subseteq \bigcup_{a \elm U}\Ext(a)$, so
  there exists a finite $V \subseteq U$ such that $[p',q']
  \subseteq \bigcup_{a \elm V}\Ext(a)$. Thus $(p',q') \cov_{\FReal}
  V$ by Lemma \ref{lem:FinCovR}. Hence $(p,q) \cov_{\FReal} U$.
  \smallskip

  \noindent(\ref{thm:SpatialityFR4} $\Rightarrow$ \ref{thm:SpatialityFR3})
  Assume that $\FReal$ is spatial, and suppose that
  $(p,q)\cov_{\FUInt_{\Ip}} U$. Let $(p',q') <_{\FReal} (p,q)$.
  Put $(u,v) = (\max(0,p'), \min(1,q'))$. Then $[u,v] \subseteq
  \bigcup_{a \elm U} \Ext(a)$, so there exists $(u',v') >_{\FReal}
  (u,v)$ such that $(u',v') \subseteq \bigcup_{a \elm U} \Ext(a)$.
  Since $\FReal$ is spatial, we have $(u',v') \cov_{\FReal} U$.
  Moreover, it is straightforward to show that $(p',q')
  \cov_{\FUInt} \left\{ (u',v')\right\}$. Therefore
  $(p,q)\cov_{\FUInt} U$.
  \smallskip

  \noindent(\ref{thm:SpatialityFR3} $\Rightarrow$
  \ref{thm:SpatialityFR1}) Immediate from
  Proposition \ref{prop:SpatImpCont}.
\end{proof}

\begin{corollary}
  $\ContUInt$ implies that any relation $s \subseteq S_{\FReal}
  \times S_{\FReal}$ which maps every $\alpha \colon \Pt{\FReal}$ to
  $\Image{s}(\alpha)  \colon  \Pt{\FReal}$  is a formal map
  $s \colon \FReal \to \FReal$.
\end{corollary} 
\begin{proof}
  By Theorem \ref{thm:SpatialityFR} and  Proposition
  \ref{prop:SpatImpCont}.
\end{proof}

\subsection*{Acknowledgements}
We thank the anonymous referees for numerous suggestions, which helped
us improve our exposition in an essential way.
The first author was funded as a Marie Curie fellow of the Istituto
Nazionale di Alta Matematica.

\end{document}